\documentclass[11pt]{article}

\usepackage{booktabs} 
\usepackage{tikz, algorithm, bbm}
\usepackage{graphicx}
\usepackage{amsthm,amsmath,amssymb}
\usepackage[paper=letterpaper, top=.9in, bottom=.75in, left=0.9in, right=0.9in, nohead, includefoot, footskip=.25in]{geometry}

\renewcommand{\b}[1]{\left[#1\right]}
\newcommand{\E}[1]{\mathbb{E}\b{#1}}
\newcommand{\Prb}[1]{\Pr\b{#1}}
\newtheorem{thm}{Theorem}[section]
\newtheorem{lemma}[thm]{Lemma}
\theoremstyle{definition}
\newtheorem{dfn}{Definition}
\newtheorem{fact}[dfn]{Fact}

\usetikzlibrary{positioning,decorations.pathmorphing,shapes,decorations.pathreplacing}
\newcommand\numberthis{\addtocounter{equation}{1}\tag{\theequation}}
\newcommand{\p}[1]{\left(#1\right)}

\newcommand{\omt}[1]{ }

\def\ynn{Y_{(n-1:n)}}
\def\xnn{X_{(n-1:n)}}
\def\xn{X_{(n:n)}}

\newcommand*{\approxident}{%
  \mathrel{\vcenter{\offinterlineskip
  \hbox{$\sim$}\vskip-.35ex\hbox{$\sim$}\vskip-.35ex\hbox{$\sim$}}}}
\def\apeq{\approxident}
\newcommand{\apeqq}[2]{\apeq_{#1;#2}}

\def\xan{X_{(\alpha n:\alpha n)}}

\def\ynko{Y_{(n-k+1:n)}}
\def\yn{Y_{(n:n)}}

\def\zm{Z_{(m:m)}}
\def\zk{Z_{(k:m)}}
\def\zmk{Z_{(m-k:m)}}
\def\zmm{Z_{(m-1:m)}}
\def\zmko{Z_{(m-k+1:m)}}
\def\zlnn{Z_{(Cn - \ln n+1:Cn)}}
\def\znc{Z_{(n(1+c):n(1+c))}}
\def\znC{Z_{(Cn:Cn)}}
\def\znko{Z_{(n(1+c)-k+1:n(1+c))}}
\def\znn{Z_{(n(1+c)-1:n(1+c))}}
\def\znjo{Z_{(\beta^{1+\delta}n(1+c)-j : \beta^{1+\delta}n(1+c))}}
\def\znbc{Z_{(n\beta^{1+\delta}(1+c):\beta^{1+\delta}n(1+c))}}

\newcommand{\pdf}[2]{f_{(#1:#2)}}
\newcommand{\cdf}[2]{F_{(#1:#2)}}
\newcommand{\cdfz}[2]{F_{Z,(#1:#2)}}

\newcommand{\ind}[1]{\mathbbm{1}_{#1}}
\def\dt{\, dt}
\def\dx{\, dx}
\def\dy{\, dy}
\def\dz{\, dz}

\def\tpd{^{-(1+\delta)}}

\date{}

\begin{document}

\title{Selection Problems in the Presence of Implicit Bias} 
\author{Jon Kleinberg\\ Cornell University
\and Manish Raghavan\\ Cornell University}

\maketitle
\begin{abstract}
  Over the past two decades, the notion of implicit bias has come to serve as an
  important component in our understanding of discrimination in
  activities such as hiring, promotion, and school admissions. Research on
  implicit bias posits that when people evaluate others -- for example, in a
  hiring context -- their unconscious biases about membership in particular
  groups can have an effect on their decision-making, even when they
  have no deliberate intention to discriminate against members of these groups.
  A growing body of experimental work has pointed to the effect that implicit
  bias can have in producing adverse outcomes.

  Here we propose a theoretical model for studying the effects of implicit bias
  on selection decisions, and a way of analyzing possible procedural remedies
  for implicit bias within this model. 
A canonical situation represented by our model is a hiring setting: a
recruiting committee is trying to choose a set of finalists to
interview among the applicants for a job, evaluating these applicants
based on their future potential, but their estimates of potential are
skewed by implicit bias against members of one group. In this model,
we show that measures such as the Rooney Rule, a requirement that at
least one of the finalists be chosen from the affected group, can not
only improve the representation of this
  affected group, but also lead to higher payoffs in absolute terms for the
  organization performing the recruiting. However, identifying the conditions
  under which such measures can lead to improved payoffs involves subtle
  trade-offs between the extent of the bias and the underlying distribution of
  applicant characteristics, leading to novel theoretical questions about order
  statistics in the presence of probabilistic side information.
\end{abstract}

\section{Introduction}
\newcommand{\xhdr}[1]{\paragraph*{\bf {#1}.}}

Over the past two decades, the notion of {\em implicit bias}
\cite{greenwald-implicit-social-cognition}
has come to provide on important perspective on 
the nature of discrimination.
Research on implicit bias argues that unconscious attitudes 
toward members of different demographic groups --- for example,
defined by gender, race, ethnicity, national origin, sexual orientation,
and other characteristics --- can have a non-trivial impact on the way
in which we evaluate members of these groups; and this in turn may
affect outcomes in employment
\cite{bertrand-emily-greg,bohnet-implicit-bias-hiring,uhlmann-implicit-bias-criteria},
education \cite{van-den-bergh-implicit-bias-teachers},
law \cite{greenwald-implicit-bias-law,jolls-sunstein-implicit-bias}, 
medicine \cite{green-implicit-bias-physicians},
and other societal institutions.

In the context of a process like hiring, implicit bias thus shifts the
question of bias and discrimination to be not just about identifying
bad actors who are intentionally discriminating, but also about the
tendency of all of us to reach discriminatory conclusions based on
the unconscious application of stereotypes.
An understanding of these issues also helps inform the design
of interventions to mitigate implicit bias ---
when essentially all of us
have a latent tendency toward
low-level discrimination, a set of broader practices may be needed
to guide the process toward the desired outcome.

\xhdr{A basic mechanism: The Rooney Rule}
One of the most basic and widely adopted mechanisms in practice
for addressing implicit bias
in hiring and selection is the {\em Rooney Rule}
\cite{collins-rooney-rule-implicit-bias}, which, roughly
speaking, requires that in recruiting for a job opening, one of the
candidates interviewed must come from an underrepresented group.
The Rooney Rule is named for a protocol adopted by the 
National Football League (NFL) in 2002 in response to widespread
concern over the low representation of African-Americans in
head coaching positions; it required that when a team is searching for a
new head coach, at least one minority candidate must be interviewed
for the position.
Subsequently the Rooney Rule has become a guideline adopted
in many areas of business
\cite{cavicchia-implicit-bias-rooney}; for example, in 2015
then-President Obama exhorted leading tech firms to use
the Rooney Rule for hiring executives, and in recent years
companies including Amazon, Facebook, Microsoft, and Pinterest
have adopted a version of the Rooney Rule requiring that
at least one candidate interviewed must be a woman or a member
of an underrepresented minority group
\cite{passariello-implicit-rooney-wsj}.
In 2017, a much-awaited set of recommendations made by
Eric Holder and colleagues to address workplace bias 
at Uber advocated for the use of the Rooney Rule
as one of its key points
\cite{covington-burling-uber-recomm,shaban-implicit-rooney-wapo}.

The Rooney Rule is the subject of ongoing debate, and one crucial aspect
of this debate is the following tension.
On one side is the argument that implicit (or explicit) bias
is preventing deserving candidates from underrepresented groups from being
fairly considered, and the Rooney Rule is providing a
force that counter-balances and partially offsets the consequences
of this underlying bias.
On the other side is the concern that if a job search process produces
a short-list of top candidates all from the majority group, it may
be because these are genuinely the strongest candidates despite
the underlying bias --- particularly if there is a shortage of
available candidates from other groups.  In this case, wholesale
use of the Rooney Rule may lead firms to consider weaker candidates
from underrepresented groups, which works against the elimination
of unconscious stereotypes. 
Of course, there are other reasons to seek diversity in recruiting
that may involve broader considerations or longer time horizons than 
just the specific applicants being evaluated;
but even these lines of argument generally incorporate 
the more local question of the effect on the set of applicants.

Given the widespread consideration of the Rooney Rule from both legal and
empirical perspectives \cite{collins-rooney-rule-implicit-bias},
it is striking that prior work has not attempted to 
formalize the inherently mathematical question that forms
a crucial ingredient in these debates:
given some estimates of the extent of bias and the prevalence
of available minority candidates, 
does the expected quality of the candidates being interviewed by a hiring
committee go up or down when the Rooney Rule is implemented?
When the bias is large and there are many minority candidates,
it is quite possible that a hiring committee's bias has caused it
to choose a weaker candidate over a stronger minority one, and
the Rooney Rule may be strengthening the pool of interviewees by
reversing this decision and swapping the stronger minority candidate in.
But when the bias is small or there
are few minority candidates, the Rule might be reversing a decision
that in fact chose the stronger applicant.

In this paper, we propose a formalization of this family of questions,
via a simplified model of selection with implicit bias,
and we give a tight analysis of the consequences of using
the Rooney Rule in this setting.
In particular, when selecting for a fixed number of slots, we 
identify a sharp threshold on the effectiveness of the Rooney Rule 
in our model that depends on three parameters: 
not just the extent of bias and the the prevalence of available
minority candidates, but a third quantity as well --- essentially, a parameter
governing the distribution of candidates' expected future job performance. 
We emphasize that our model is deliberately stylized, to 
abstract the trade-offs as cleanly as possible.
Moreover, in interpreting these results, we emphasize a point noted above,
that there are other reasons to consider using the Rooney Rule beyond the
issues that motivate this particular formulation;
but an understanding of the trade-offs in our model seems informative 
in any broader debate about such hiring and selection measures.

We now describe the basic ingredients of our model, followed
by a summary of the main results.

\subsection{A Model of Selection with Implicit Bias}

Our model is based on the following scenario.
Suppose that a hiring committee is trying to fill an open job position,
and it would like to choose the $k \geq 2$ best candidates as {\em finalists}
to interview from among a large set of applicants.
We will think of $k$ as a small constant, and indeed
most of the subtlety of the question already arises 
for the case $k = 2$, when just two finalists must be selected.

\xhdr{$X$-candidates and $Y$-candidates}
The set of all applicants is partitioned into two groups $X$ and $Y$,
where we think of $Y$ as the majority group, and $X$ as a minority
group within the domain that may be subject to bias.
For some positive real number $\alpha \leq 1$ and a natural number $n$, 
there are $n$ applicants 
from group $Y$ and $\alpha n$ applicants from group $X$.
If a candidate $i$ belongs to $X$, we will refer to them as an
{\em $X$-candidate}, and if 
$i$ belongs to $Y$, we will refer to them as a
{\em $Y$-candidate}.
(The reader is welcome, for example, to think of the setting
of academic hiring, with $X$ as candidates from a group 
that is underrepresented in the field,
but the formulation is general.)

Each candidate $i$ has a (hidden) numerical value that we call
their {\em potential},
representing their future performance over the course of their career.
For example, in faculty hiring, we might think of the potential
of each applicant in terms of some numerical proxy
like their future lifetime citation count
(with the caveat that any numerical measure will of course be an imperfect
representation).
Or in hiring executives, the potential of each applicant could be
some measure of the revenue they will bring to the firm.

We assume that 
there is a common distribution $Z$ that these numerical potentials come from:
each potential is an independent draw from $Z$.  (Thus, the
applicants can have widely differing values for
their numerical potentials; they just
arise as draws from a common distribution.)
For notational purposes, when $i$ is an $X$-candidate, we write
their potential as $X_i$, and when $j$ is a $Y$-candidate, we write
their potential as $Y_j$.
We note an important modeling decision in this formulation: we are
assuming that all $X_i$ and all $Y_j$ values come from this same
distribution $Z$.
While it is also of interest to consider the case in which the numerical
potentials of the two groups $X$ and $Y$ are drawn from different
group-specific distributions, we focus on the case of identical 
distributions for two reasons.  First, there are many settings
where differences between the underlying distributions for 
different groups appear to be small compared
to the bias-related effects we are seeking to measure; and second,
in any formal analysis of bias between groups, the setting in which
the groups begin with identical distributions is arguably the
first fundamental special case that needs to be understood.

In the domains that we are considering --- hiring executives,
faculty members, athletes, performers --- there is a natural functional form
for the distribution $Z$ of potentials, and this is the family
of {\em power laws} (also known as {\em Pareto distributions}), with
$\Prb{Z \geq t} = t^{-(1 + \delta)}$ and support $[1, \infty)$ for a
fixed $\delta > 0$.
Extensive empirical work has argued that the distribution of
individual output in a wide range of creative professions can
be approximated by power law distributions with small
positive values of $\delta$ \cite{clauset-power-law-survey}.
For example, the distribution of lifetime citation counts
is well-approximated by a power law, as are the lifetime
downloads, views, or sales by performers, authors, and other artists.
In the last part of the paper, we also consider the case in which
the potentials are drawn from a distribution with bounded support,
but for most of the paper we will focus on power laws.

\xhdr{Selection with Bias}
Given the set of applicants,
the hiring committee would like to choose $k$ {\em finalists} to interview.
The {\em utility} achieved by the committee is the sum of the potentials
of the $k$ finalists it chooses; the committee's goal is to maximize
its utility.\footnote{Since our goal is to model processes like
the Rooney Rule, which apply to the selection of finalists for
interviewing, rather than to the hiring decision itself, we treat
the choice of $k$ finalists as the endpoint rather than
modeling the interviews that subsequently ensue.}

If the committee could exactly evaluate the potential of each 
applicant, then it would have a straightforward way to maximize
the utility of the set of finalists: simply sort all applicants
by potential, and choose the top $k$ as finalists.
The key feature of the situation we would like to capture, however,
is that the committee is biased in its evaluations; 
we look for a model that incorporates this bias as cleanly as possible.

Empirical work in some of our core motivating settings --- such as
the evaluation of scientists and faculty candidates --- indicates
that evaluation committees often systematically downweight female
and minority candidates of a given level of achievement, 
both in head-to-head comparisons and in ranking using numerical scores
\cite{wenneras-implicit-bias-peer-rev}.
It is thus natural to model the hiring committee's evaluations as follows:
they correctly estimate the potential of a $Y$-applicant $j$ at 
the true value $Y_j$, but they estimate the potential of an
$X$-applicant $i$ at a reduced value $\tilde{X_i} < X_i$.
They then rank candidates by these values $\{Y_j\}$ and $\{\tilde{X_i}\}$,
and they choose the top $k$ according to this biased ranking.

For most of the paper, we focus on the case of
{\em multiplicative bias}, in which $\tilde{X_i} = X_i / \beta$
for a bias parameter\footnote{When $\beta = 1$, the ranking has no bias.} $\beta
> 1$. This is a reasonable approximation to empirical data from human-subject
studies \cite{wenneras-implicit-bias-peer-rev}; and moreover, for power law
distributions this multiplicative form is in a strong sense the ``right''
parametrization of the bias, since biases that grow either faster or slower than
multiplicatively have a very simple asymptotic behavior in the power law case.

In this aspect of the model, as in others, we seek the cleanest
formulation that exposes the key underlying issues;
for example, it would be an interesting extension to consider versions in 
which the estimates for each individual are perturbed by random noise.
A line of previous work
\cite{
braverman-sorting-noisy,
feige-computing-noisy,
fu-micro-foundations} 
has analyzed models of ranking under noisy perturbations; 
while our scenario is quite different in
that the entities being ranked are partitioned into a fixed set of
groups with potentially different levels of bias and noise,
it would be natural to see if 
these techniques could potentially be extended to handle
noise in the context of implicit bias. 

\subsection{Main Questions and Results}

This then is the basic model in which we analyze interventions 
with the structure of the Rooney Rule:
(i) a set of $n$ $Y$-applicants and $\alpha n$ $X$-applicants each
have an independent future potential drawn from a power law 
distribution; 
(ii) a hiring committee ranks the applicants according to a sorted
order in which each $X$-applicant's potential is divided down by $\beta > 1$,
and chooses the top $k$ in this ordering as {\em finalists};
and (iii) the hiring committee's {\em utility} is the sum of
the potentials of the $k$ finalists.

Qualitatively, the motivation for the Rooney Rule in such settings
is that hiring committees are either unwilling or unable to
reasonably correct for their bias in performing such rankings, and
therefore cannot be relied on to interview $X$-candidates on their own.
The difficulty in removing this skew from such evaluations is
a signature aspect of phenomena around implicit bias.

The decision to impose the Rooney Rule is 
made at the outset, before the actual values of the potentials 
$\{Y_j\}$ and $\{\tilde{X_i}\}$ are materialized.
All that is known at the point of this initial decision
to use the Rule or not are the parameters of the domain:
the bias $\beta$, the relative abundance of $X$-candidates $\alpha$,
the power law exponent $1 + \delta$, and the number of finalists
to be chosen $k$.
The question is: as a function of these parameters,
will the use of the
Rooney Rule produce a positive or negative expected change in utility,
where the expectation is taken over the random draws of applicant values?
We note that one could instead ask about the probability that the Rooney Rule
produces a positive change in utility as opposed to the expected change; in
fact, our techniques naturally extend to characterize not only the expected
change, but the probability that this change is positive, as we will show in
Section~\ref{sec:pl}.

\begin{figure}[t]
\begin{center}
\includegraphics[width=.80\textwidth]{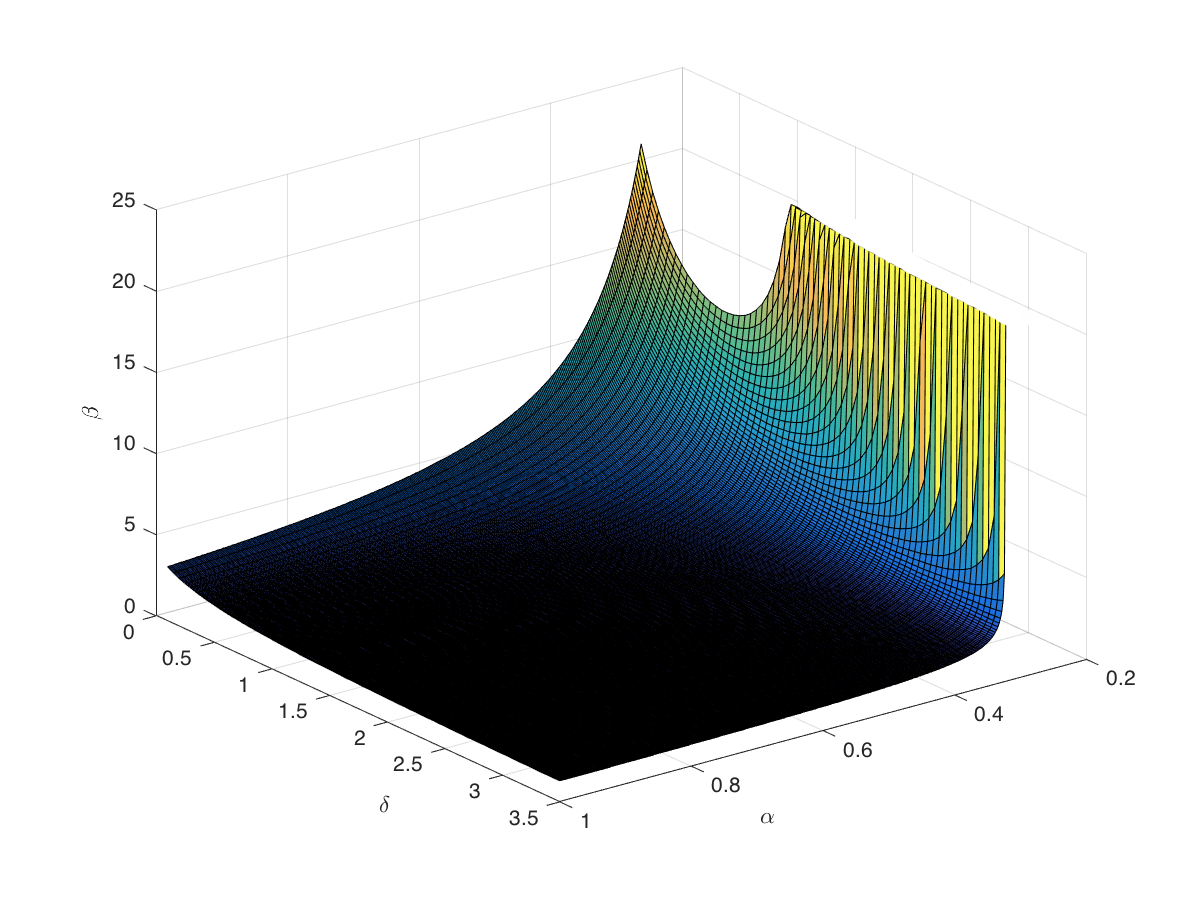}
\caption{
\label{fig:phi}
Fixing $k = 2$, the $(\alpha,\beta,\delta)$ values for
which the Rooney Rule produces a positive expected change for
sufficiently large $n$ lie above a surface (depicted in the figure)
defined by the function $\phi_2(\alpha,\beta,\delta) = 1$.
}
\end{center}
\end{figure}

Our model lets us make precise the trade-off in
utility that underpins the use of the Rooney Rule.
If the committee selects an $X$-candidate on its own --- even
using its biased ranking --- then their choice already
satisfies the conditions of the Rule.
But if all $k$ finalists are $Y$-candidates, then the Rooney Rule
requires that the committee replace the lowest-ranked of
these finalists $j$ with the highest-ranked $X$-candidate $i$.
Because $i$ was not already a finalist, we know that 
$\tilde{X_i} = X_i / \beta < Y_j$.
But to see whether this yields a positive change in utility,
we need to understand which of $X_i$ or $Y_j$ has a larger expected value,
conditional on the information contained in the committee's 
decision, that $X_i / \beta < Y_j$.

Our main result is an exact characterization of when the 
Rooney Rule produces a positive expected change in terms of the four
underlying parameters, showing that it non-trivially depends on all four. For
the following theorem, and for the remainder of the paper, 
we assume $0 < \alpha \leq 1$, $\beta > 1$, and $\delta > 0$. We begin with the
case where $k = 2$.

\begin{thm}
  For $k = 2$ and sufficiently large $n$, the Rooney Rule produces a positive
  expected change if and only if $\phi_2(\alpha, \beta, \delta) > 1$ where
  \begin{equation}
    \phi_2(\alpha, \beta, \delta) = \frac{\alpha^{1/(1+\delta)}\b{1 -
    (1+c^{-1})^{-\delta/(1+\delta)} \b{1 + \frac{\delta}{1+\delta}
    (1+c)^{-1}}}}{\frac{\delta}{1+\delta} (1+c)^{-1-\delta/(1+\delta)}}
  \end{equation}
  and $c = \alpha \beta^{-(1+\delta)}$. Moreover, $\phi_2(\alpha, \beta,
  \delta)$ is increasing in $\beta$, so for fixed $\alpha$ and $\delta$ there
  exists $\beta^*$ such that $\phi_2(\alpha, \beta, \delta) > 1$ if and only if
  $\beta > \beta^*$.
  \label{thm:characterization2}
\end{thm}

Thus, we have an explicit characterization for when the Rooney Rule produces
positive expected change. The following theorem extends this to larger values of
$k$.

\begin{thm}
There is an explicit
function $\phi_k(\alpha,\beta,\delta)$ such
that the Rooney Rule produces a positive expected change, for
$n$ sufficiently large and $k = O(\ln n)$, 
if and only if $\phi_k(\alpha,\beta,\delta) > 1$.
\label{thm:characterization}
\end{thm}

Interestingly, even for larger values of $k$, there are parts of the parameter
space for which the Rooney Rule produces a positive expected change
and parts for which the Rooney Rule produces a negative expected change,
independent of the number of applicants $n$.

Figure \ref{fig:phi} depicts a view of the function $\phi_2$, by
showing the points in three-dimensional
$(\alpha,\beta,\delta)$ space for which $\phi$ takes the value $1$.
The values for which the Rooney Rule produces a positive expected change for
sufficiently large $n$ lie above this surface. 

The surface in Figure \ref{fig:phi} is fairly complex, 
and it displays unexpected non-monotonic behavior. 
For example, on certain regions of fixed $(\alpha,\beta)$,
it is non-monotonic in $\delta$, a fact which is not a priori obvious:
there are choices of $\alpha$ and $\beta$ for which the Rooney
Rule produces a positive expected change at certain ``intermediate''
values of $\delta$, but not at values of $\delta$ that are
sufficiently smaller or sufficiently larger.
Moreover, there exist $(\alpha, \delta)$ pairs above which the surface
does not exist. (One example in Figure \ref{fig:phi} occurs at
$\alpha \approx 0.3$ and $\delta \approx 3$). 
Characterizing the function $\phi$ and its level set $\phi = 1$
is challenging, and it is noteworthy that the complexity of this
function is arising from our relatively bare-bones formulation of the trade-off
in the Rooney Rule; this suggests 
the function and its properties are capturing something
inherent in the process of biased selection. 

One monotonocity result we are able to establish for the function
$\phi$ is the following, showing that for fixed $(\alpha,\beta,\delta)$,
increasing the number
of positions can't make the Rooney Rule go from beneficial to harmful.

\begin{thm}
For sufficiently large $n$ and $k = O(\ln n)$, if the Rooney Rule produces a
positive expected change at a given number of finalists $k$, it also produces a
positive expected change when there are $k + 1$ finalists (at the same
$(\alpha,\beta,\delta)$).
\label{thm:monotone}
\end{thm}

We prove these theorems through an analysis of the {\em order statistics}
of the underlying power law distribution.
Specifically, if we draw $m$ samples from the power law $Z$
and sort them in ascending order from lowest to highest, then
the $\ell^{\rm th}$ item in the sorted list is a random variable
denoted $Z_{(\ell:m)}$.
To analyze the effect of the Rooney Rule, we are comparing
$Y_{(n-k+1:n)}$ with $X_{(\alpha n : \alpha n)}$.
Crucially, we are concerned with their expected values
conditional on the fact that the committee chose the $k^{\rm th}$-ranked
$Y$-candidate over the top-ranked $X$-candidate, implying 
as noted above that $X_{(\alpha n : \alpha n)} / \beta < Y_{(n-k+1:n)}$.
The crucial comparison is therefore between 
$\E{Y_{(n-k+1:n)} | \xan < \beta Y_{(n-k+1:n)}}$ and 
$\E{\xan | \xan < \beta Y_{(n-k+1:n)}}$.
Order statistics conditional on this type of side information turn out to 
behave in complex ways, and hence the core of the analysis is in dealing
with these types of conditional order statistics for power law distributions.

More generally, given the ubiquity of power law distributions
\cite{clauset-power-law-survey}, we find it surprising how little is
known about how their order statistics behave qualitatively. In this
respect, the techniques we provide may prove to be independently
useful in other applications. For example, we develop a tight
asymptotic characterization of the expectations of order statistics
from a power law distribution that to our knowledge is novel.

We also note that 
although our results are expressed for sufficiently large $n$,
the convergence to the asymptotic behavior happens very quickly as
$n$ grows; to handle fixed values of $n$, we need only modify the
bounds by correction terms that grow like $\p{1 \pm O\p{\frac{(\ln n)^2}{n}}}$.
In particular, the errors in the asymptotic analysis are small once
$n$ reaches 50, which is reasonable for settings in which a
job opening receives many applications.

\xhdr{Estimating the level of bias $\beta$}
The analysis techniques we develop for proving Theorem
\ref{thm:characterization} 
can also be used for related problems in this model.
A specific question we are able to address is the problem of estimating
the amount of bias from a history of hiring decisions.

In particular, suppose that over $m$ years the hiring committee
makes one offer per year; in $N$ of the $m$ years this offer goes
to an $X$-candidate, and in $m - N$ of the $m$ years this offer
goes to a $Y$-candidate.
Which value of the bias parameter $\beta$ maximizes the probability
of this sequence of observations?

We provide a tight characterization of the solution to this question,
finding again that it
depends not only on $\alpha$ (in this case, the sequence of $\alpha$
values for each year), but also on the power law exponent $1 + \delta$.
The solution has a qualitatively natural structure, and produces $\beta = 1$
(corresponding to no bias)
as the estimate when the fraction of $X$-candidates hired over the
$m$ years is equal to the expected number that would be hired under
random selection.

\xhdr{Generalizations to other distributions}
Finally, at the end of the paper we consider how to adapt our approach
for classes of distributions other than power laws.
A different category of distributions that can be motivated by the
considerations discussed here is the set of bounded distributions,
which take values only over a finite interval.
Just as power laws are characteristic of the performance of
employees in certain professions, bounded distributions are 
appropriate when there are absolute constraints on the maximum effect
a single employee can have.

Moreover, bounded distributions are also of interest because they
contain the uniform distribution on $[0,1]$ as a special case.
We can think of this special case as describing an instance in 
which each candidate is associated with their {\em quantile}
(between $0$ and $1$) in a ranking of the entire population,
and the bias then operates on this quantile value,
reducing it in the case of $X$-candidates.

For bounded distributions, we can handle much more general forms
for the bias --- essentially, any function that reduces the values
$X_i$ strictly below the maximum of the distribution (for instance, a bias that
always prefers a $Y$-candidate to an $X$-candidate when they are within
some $\varepsilon$ of each other). When $k = 2$ and 
there are equal numbers of $X$-candidates and $Y$-candidates,
we show that for any bounded distribution and any such bias, 
the Rooney Rule produces a positive expected change in utility
for all sufficiently large $n$.

\subsection{An Illustrative Special Case: Infinite Bias} \label{sec:inf_bias}

To illustrate some of the basic considerations that go into
our analysis and its interpretation, we begin with a
simple special case that we can think of as ``infinite bias'' ---
the committee deterministically ranks every $Y$-candidate above
every $X$-candidate.  This case already exhibits structurally rich
behavior, although the complexity is enormously less than the
case of general $\beta$.  We also focus here on $k = 2$. In terms of
Figure~\ref{fig:phi}, we can visualize the infinite bias case as if we are
looking down at the plot from infinitely high up; 
thus, reasoning about infinite bias
amounts to determining which parts of the 
$(\alpha, \delta)$ plane are covered by the
surface $\phi_2(\alpha, \beta, \delta) = 1$.

With infinite bias, the committee is guaranteed to choose 
the two highest-ranked $Y$-candidates in the absence of an intervention;
with the Rooney Rule, the committee will choose the highest-ranked
$Y$-candidate and the highest-ranked $X$-candidate.
As we discuss in the next section, for power law distributions
with exponent $1 + \delta$, if $z^*$ is the expected maximum 
of $n$ draws from the distribution, then (i) the expected value of
the second-largest of the $n$ draws is $\frac{\delta}{(1 + \delta)} z^*$;
and (ii) the expected maximum of $\alpha n$ draws from the 
distribution is asymptotically $\alpha^{1/(1 + \delta)} z^*$.

This lets us directly evaluate the utility consequences of the
intervention.  If there is no intervention, the utility of the
committee's decision will be $\left(1 + \frac{\delta}{1 + \delta}\right)
z^*$, and if the Rooney Rule is used, the utility of the
committee's decision will be $(1 + \alpha^{1/(1 + \delta)}) z^*$.
Thus, the Rooney Rule produces positive expected change in utility
if and only if $\alpha^{1/(1 + \delta)} > \frac{\delta}{(1 + \delta)}$;
that is, if and only if
$\alpha > \left(\frac{\delta}{1 + \delta}\right)^{1 + \delta}$.

In addition to providing a simple closed-form expression for 
when to use the Rooney Rule in this setting,
the condition itself leads to some counter-intuitive consequences.
In particular, the closed-form expression for the condition 
makes it clear that {\em for every} $\alpha > 0$, 
there exists a sufficiently small $\delta > 0$ so that when the
distribution of applicant potentials is a power law with exponent $1 + \delta$,
using the Rooney Rule produces the higher expected utility.
In other words, with a power law exponent close to 1, it's a better
strategy to commit one of the two offers to the $X$-candidates,
even though they form an extremely small fraction of the population.

This appears to come perilously close to contradicting the following
argument.  We can arbitrarily divide the $Y$-candidates into two
sets $A$ and $B$ of $n/2$ each; and if $\alpha < 1/2$, each of 
$A$ and $B$ is larger than the set of all $X$-candidates.
Let $a^*$ be the top candidate in $A$ and $b^*$ be the top candidate in $B$.
Each of $a^*$ and $b^*$ has at least the expected value of the top
$X$-candidate, and moreover, one of them is the top 
$Y$-candidate overall.  So how can it be that choosing $a^*$ and $b^*$
fails to improve on the result of using the Rooney Rule?

The resolution is to notice that using the Rooney Rule still involves hiring
the {\em top} $Y$-candidate.  So it's not that the Rooney Rule chooses one
of $a^*$ or $b^*$ at random, together with the top $X$-candidate.
Rather, it chooses the {\em better} of $a^*$ and $b^*$, along with the top
$X$-candidate.
The real point is that
power law distributions have so much probability in the 
tail of the distribution
that the best person among a set of $\alpha n$ 
can easily have a higher expected
value than the second-best person among a set of $n$,
even when $\alpha$ is quite small.
This is a key property of power law distributions that
helps explain what's happening both in this example and in our analysis.

\subsection{A Non-Monotonicity Effect}
As noted above, much of the complexity in the analysis arises from
working with expected values of random variables
conditioned on the outcomes of certain biased comparisons.  
One might hope that expected values conditional on these types of comparisons
had tractable properties that facilitated the analysis, but
this is not the case; 
in fact, these conditional expectations exhibit 
some complicated and fairly counter-intuitive behavior.
To familiarize the reader with some of these phenomena --- both as
preparation for the subsequent sections, but also as an interesting
end in itself --- we offer the following example.

Much of our analysis involves quantities like
$\E{X | X > \beta Y}$ --- the conditional expectation of $X$, given
that it exceeds some other random variable $Y$ multiplied by a bias parameter.
(We will also be analyzing the version in which the inequality
goes in the other direction, 
but we'll focus on the current expression for now.)
If we choose $X$ and $Y$ as independent random variables both drawn from
a distribution $Z$, and then view the conditional expectation as a function
just of the bias parameter $\beta$, what can we say about the properties of
this function $f(\beta) = \E{X | X > \beta Y}$?

Intuitively we'd expect $f(\beta)$ to be monotonically increasing
in $\beta$ --- indeed, as $\beta$ increases, we're putting a stricter
lower bound on $X$, and so this ought to raise the conditional expectation
of $X$.

The surprise is that this is not true in general; we can construct
independent random variables $X$ and $Y$ for which $f(\beta)$ is not
monotonically increasing.
In fact, the random variables are very simple: we can have each of $X$ and $Y$
take values independently and uniformly from the finite set
$\{1, 5, 9, 13\}$.
Now, the event $X > 2Y$ consists of four possible pairs of $(X,Y)$ values:
(5,1), (9,1), (13,1), and (13,5).  
Thus, $f(2) = \E{X | X > 2 Y} = 10$.
In contrast, the event $X > 3Y$ consists of three possible pairs of
$(X,Y)$ values: (5,1), (9,1), and (13,1).  Thus, $f(3) = 9$, which
is a smaller value, despite the fact that $X$ is required to be 
a larger multiple of $Y$.

The surprising content of this example has a fairly sharp formulation
in terms of a story about recruiting.
Suppose that two academic departments, Department $A$ and Department $B$,
both engage in hiring each year.  
In our stylized setting, each interviews one $X$-candidate and 
one $Y$-candidate each year, and hires one of them.
Each candidate comes from the uniform distribution on $\{1, 5, 9, 13\}$.
Departments $A$ and $B$ are both biased in their hiring: 
$A$ only hires the $X$-candidate in a given year if they're more than
twice as good as the $Y$-candidate, while
$B$ only hires the $X$-candidate in a given year if they're more than
three times as good as the $Y$-candidate.

Clearly this bias hurts the average quality of both departments, $B$ more so
than $A$.  But you might intuitively expect that at least if you looked
at the $X$-candidates that $B$ has actually hired, they'd be of higher average 
quality than the $X$-candidates that $A$ has hired --- simply because
they had to pass through a stronger filter to get hired.
In fact, however, this isn't the case: despite the fact that $B$ imposes
a stronger filter, the calculations performed above for this example
show that the average quality of the $X$-candidates $B$ hires is $9$,
while the average quality of the $X$-candidates $A$ hires is $10$.

This non-monotonicity property shows that the conditional expectations
we work with in the analysis can be pathologically behaved for arbitrary
(even relatively simple) distributions.
However, we will see that with power law distributions we are able ---
with some work --- to avoid these difficulties; and part of our
analysis will include a set of explicit monotonicity results.

\section{Biased Selection with Power Law Distributions} \label{sec:pl}
Recall that for a random variable $Z$, we use $Z_{(\ell:m)}$
to denote the $\ell^{\rm th}$ order statistic in $m$ draws from $Z$:
the value in position $\ell$ when we sort $m$ independent
draws from $Z$ from lowest to highest. 
Recall also that
when selecting $k$ finalists, the Rooney Rule improves expected
utility exactly when
\[
  \E{\xan - \ynko | \xan < \beta \ynko} > 0.
\]
Using linearity of expectation and the fact that $\Prb{A|B} \Prb{B} = \Prb{A
\cdot \ind{B}}$, this is equivalent to
\begin{equation}
  \frac{\E{\xan \cdot \ind{\xan < \beta \ynko}}}{\E{\ynko \cdot \ind{\xan <
  \beta \ynko}}} > 1.
  \label{eq:decision}
\end{equation}
We will show an asymptotically tight characterization of the
tuples of parameters $(k,\alpha,\beta,\delta)$ for which this condition holds,
up to an error term on the order of
$O\p{\frac{(\ln n)^2}{n}}$. In order to better understand the terms
in~\eqref{eq:decision}, we begin with some necessary background.

\subsection{Preliminaries}
\begin{fact}
  Let $\pdf{p}{m}$ and $\cdf{p}{m}$ be, respectively, the 
  probability density function and cumulative distribution function of the 
  $p^{\rm th}$ order statistic
  out of $m$ draws from the power law distribution with
  parameter $\delta$. Using definitions from \cite{order-statistics},
  \begin{align*}
    \pdf{p}{m}(x)
    &= (1+\delta) (m-p+1) \binom{m}{p-1} \p{1-x^{-(1+\delta)}}^{p-1}
    \p{x^{-(1+\delta)}}^{m-p+1} x^{-1}
  \end{align*}
  and
  \[
    \cdf{p}{m}(x)
    = \sum_{j=p}^m \binom{m}{j} \p{1-x^{-(1+\delta)}}^j
    \p{x^{-(1+\delta)}}^{m-j}.
  \]
\end{fact}
\begin{dfn}
  We define 
  \[
    \Gamma(a) = \int_0^\infty t^{a-1} e^{-t} \dt.
  \]
  $\Gamma(\cdot)$ is considered the continuous relaxation of the factorial, and
  it satisfies
  \[
    \Gamma(a+1) = a\Gamma(a).
  \]
  If $a$ is a positive integer, $\Gamma(a+1) = a!$.
  Furthermore, $\Gamma(a) > 1$ for $0 < a < 1$ and $\Gamma(a) < 1$ for $1 < a <
  2$.
\end{dfn}

\subsection{The Case where $k=2$}
For simplicity, we begin with the case where we're
selecting $k = 2$ finalists.
In this section, we will make several approximations, growing tight
with large $n$, that we treat formally in Appendices~\ref{app:missing_pl}
and~\ref{app:pl_proofs}. This section is intended to demonstrate the techniques
needed to understand the condition~\eqref{eq:decision}. In the case where $k=2$,
always selecting an $X$-candidate increases expected utility if and only if
\begin{equation}
  \frac{\E{\xan \cdot \ind{\xan < \beta \ynn}}}{\E{\ynn \cdot \ind{\xan <
    \beta \ynn}}} > 1.
  \label{eq:decision_2}
\end{equation}
Theorems~\ref{thm:cond_x_k} and~\ref{thm:cond_y_k} in
Appendix~\ref{app:pl_proofs} give tight approximations to these quantities;
here, we provide an outline for how to find them. For the sake of exposition, we'll
only show this for the denominator in this section, which is slightly simpler to
approximate. We begin with
\[
  \E{\ynn \cdot \ind{\xan < \beta \ynn}}
  = \int_1^\infty y \pdf{n-1}{n}(y) \cdf{\alpha n}{\alpha n}(\beta y) \dy.
\]
Letting $c = \alpha \beta^{-(1+\delta)}$, we can use Lemma~\ref{lem:y_approx_c}
and some manipulation to approximate this by
\begin{align*}
  (1+\delta) n(n-1) \int_1^\infty \p{1 - y^{-(1+\delta)}}^{n(1+c)-2}
  \p{y^{-(1+\delta)}}^2 \dy.
\end{align*}
Conveniently, the function being integrated is (up to a constant factor) $y
\cdot \pdf{n(1+c)-1}{n(1+c)}(y)$, i.e. $y$~times the probability 
density function of the second-highest
order statistic from $n(1+c)$ samples. Since
\begin{align*}
  \E{\znn}
  &= \int_1^\infty z \pdf{n(1+c)-1}{n(1+c)}(z) \dz \\
  &= (1+\delta) n(1+c) (n(1+c)-1) \int_1^\infty \p{1 -
    z^{-(1+\delta)}}^{n(1+c)-2} \p{z^{-(1+\delta)}}^2 \dz,
\end{align*}
we have
\[
  \E{\ynn \cdot \ind{\xan < \beta \ynn}} \approx \frac{1}{(1+c)^2} \E{\znn}.
\]
Then, we can use Lemmas~\ref{lem:os_approx} and~\ref{lem:recurrence} to get
$\E{\znn} \approx (1+c)^{1/(1+\delta)} \E{\ynn}$, meaning that
\begin{equation}
  \label{eq:y_cond_2}
  \E{\ynn \cdot \ind{\xan < \beta \ynn}} \approx (1+c)^{-(1+\delta/(1+\delta))}
  \E{\ynn}.
\end{equation}
For the numerator of~\eqref{eq:decision_2}, a slightly more involved calculation yields
\begin{equation}
  \E{\xan \cdot \ind{\xan < \beta \ynn}} \approx \E{\xan}\b{1 -
  (1+c^{-1})^{-\delta/(1+\delta)} \b{1 + \frac{\delta}{1+\delta}
  (1+c)^{-1}}}.
  \label{eq:x_cond_2}
\end{equation}

By Lemmas~\ref{lem:os_approx} and~\ref{lem:recurrence}, $\E{\xan} \approx
\Gamma\p{\frac{\delta}{1+\delta}} (\alpha n)^{1/(1+\delta)}$ and $\E{\ynn}
\approx \Gamma\p{1+\frac{\delta}{1+\delta}}
n^{1/(1+\delta)}$. Recall that, up to the approximations we made, the Rooney
Rule improves utility in expectation if and only if the ratio
between~\eqref{eq:x_cond_2} and~\eqref{eq:y_cond_2} is larger than 1. Therefore,
the following theorem holds:
\begin{thm} \label{thm:r_2}
  For sufficiently large $n$, the Rooney Rule with $k=2$ improves utility in
  expectation if and only if
  \begin{equation}
    \frac{\alpha^{1/(1+\delta)}\b{1 -
    (1+c^{-1})^{-\delta/(1+\delta)} \b{1 + \frac{\delta}{1+\delta}
    (1+c)^{-1}}}}{\frac{\delta}{1+\delta} (1+c)^{-1-\delta/(1+\delta)}} > 1.
    \label{eq:condition_2}
  \end{equation}
  where $c = \alpha \beta^{-(1+\delta)}$.
\end{thm}
Note that in the limit as $\beta \to
\infty$, $c \to 0$, and the entire expression goes to
$\alpha^{1/(1+\delta)}(1+\delta)/\delta$, as noted in
Section~\ref{sec:inf_bias}.
Although the full expression in the statement of Theorem \ref{thm:r_2} 
is complex, it can be directly
evaluated, giving a tight characterization of when the Rule yields increased
utility in expectation.

With this result, we could ask for a fixed $\alpha$ and $\delta$ how to
characterize the set of $\beta$ such that the condition in~\eqref{eq:condition_2} holds. 
In fact, we can show that this expression 
is monotonically increasing in $\beta$.
\begin{thm} \label{thm:r_2_increasing}
  The left hand side of~\eqref{eq:condition_2} is decreasing in $c$ and
  therefore increasing in $\beta$. 
  Hence for fixed $\alpha$ and $\delta$ 
  there exists $\beta^*$ such that~\eqref{eq:condition_2} holds if 
  and only if $\beta > \beta^*$.
\end{thm}

\xhdr{Non-monotonicity in $\delta$} From Theorem~\ref{thm:r_2}, we can gain some
intuition for the non-monotonicity in $\delta$ shown in Figure~\ref{fig:phi}.
For $\alpha < e^{-1}$, we can show that even with infinite bias, the Rooney Rule
has a negative effect on utility for sufficiently large $\delta$. Intuitively,
this is because the condition for positive change with infinite bias is $\alpha
> \p{\frac{\delta}{1+\delta}}^{1+\delta}$, which can be written as $\alpha >
\p{1-\frac{1}{d}}^d$ for $d = 1+\delta$. Since this converges to $e^{-1}$ from
below, for sufficiently large $\delta$ and $\alpha < e^{-1}$, 
we have $\alpha < \p{\frac{\delta}{1+\delta}}^{1+\delta}$. 
On the other hand, as $\delta \to 0$,
the Rooney Rule has a more negative effect on utility. For instance, $\phi_2(.3,
10, 1) > 1$ but $\phi_2(.3, 10, .5) < 1$. Intuitively, this non-monotonicity
arises from the fact that for large $\delta$ and small $\alpha$, the Rooney Rule
always has a negative impact on utility, while for very small $\delta$, samples
are very far from each other, meaning that the bias has less effect on the
ranking.

\subsection{The General Case}
We can extend these techniques to handle larger values of $k$.
For $k \in [n]$, we define
\[
  r_k(\alpha, \beta, \delta) = \frac{\E{\xan | \xan < \beta \ynko}}{\E{\ynko |
  \xan < \beta \ynko}} = \frac{\E{\xan \cdot \ind{\xan < \beta \ynko}}}{\E{\ynko \cdot
  \ind{\xan < \beta \ynko}}}.
\]
We can see that the Rooney Rule improves expected utility when selecting $k$
candidates if and only if $r_k > 1$. While $r_k$ depends on $n$, we will show
that it is a very weak dependence: for small $k$, as $n$ increases, 
$r_k$ converges to a function of
$(\alpha, \beta, \delta, k)$ up to a $1 + O((\ln n)^2/n)$ multiplicative factor.
To make this precise, we define the following notion of asymptotic equivalence:
\begin{dfn}
  For nonnegative functions $f(n)$ and $g(n)$, define
  \[
    f(n) \apeq g(n)
  \]
  if and only if there exist $a > 0$ and $n_0 > 0$ such that
  \[
    \frac{f(n)}{g(n)} \le 1 + \frac{a (\ln n)^2}{n} ~~~~ \text{and} ~~~~
    \frac{g(n)}{f(n)} \le 1 + \frac{a (\ln n)^2}{n}
  \]
  for all $n \ge n_0$. In other words, $f(n) = g(n) \p{1 \pm O\p{\frac{(\ln
  n)^2}{n}}}$. When being explicit about $a$ and $n_0$, we'll write $f(n)
  \apeqq{a}{n_0} g(n)$.
\end{dfn}
Appendix~\ref{app:equiv} contains a series of lemmas 
establishing how to rigorously manipulate
equivalences of this form. Now, we formally define a tight
approximation to $r_k$, which serves as an expanded restatement of
Theorem~\ref{thm:characterization} from the introduction.
\begin{thm} \label{thm:rk}
  For $k \in [n]$, define
  \begin{equation}
    \phi_k(\alpha, \beta, \delta) = \frac{\alpha^{1/(1+\delta)} c^{\delta/(1+\delta)}
    (1+c)^{k - 1}}{\binom{k-1-\frac{1}{1+\delta}}{k-1}} \b{(1+c^{-1})^{\delta/(1+\delta)} -
    \sum_{j=0}^{k-1} \binom{j-\frac{1}{1+\delta}}{j} (1 + c)^{-j}}
    \label{eq:phi_def}
  \end{equation}
  where $c = \alpha \beta^{-(1+\delta)}$. Note that $\phi_k$ does not depend on
  $n$. When $(\alpha, \beta, \delta)$ are fixed, we will simply write this as
  $\phi_k$. For $k \le ((1-c^2)\ln n)/2$, we have 
  \begin{align*}
    r_k \apeq \phi_k,
  \end{align*}
  and therefore the Rooney Rule improves expected utility for sufficiently
  large $n$ if and only if $\phi_k > 1$.
\end{thm}
This condition tightly characterizes when the Rooney Rule improves expected
utility, and its asymptotic nature in $n$ becomes accurate even for 
moderately small $n$: for example, when $n = 50$, 
the error between $r_k$ and $\phi_k$ is around $1\%$ for reasonable choices of
$(\alpha, \beta, \delta)$.

\xhdr{Increasing $k$}
Consider the scenario in which we're selecting $k$ candidates, and for the given
parameter values, the Rooney Rule improves our expected utility. If we were to
instead select $k+1$ candidates, should we still be reserving a spot for an
$X$-candidate? Intuitively, as $k$ increases, the Rule is less likely to change
our selections, since we're more likely to have already chosen an $X$-candidate;
however, it is not a priori obvious whether increasing $k$ should make it better
for us to use the Rooney Rule (because we have more slots, so we're losing less
by reserving one) or worse (because as we take more candidates, we stop needing
a reserved slot).

In fact, we can apply Theorem~\ref{thm:rk} to understand how $r_k$ changes with
$k$. The following theorem, proven in Appendix~\ref{app:pl_proofs}, is
an expanded restatement of Theorem~\ref{thm:monotone}, showing that
if the Rooney Rule yields an improvement in expected quality when selecting $k$
candidates, it will do so when selecting $k+1$ candidates as well.
\begin{thm} \label{thm:rk_increasing}
  For $k \le ((1-c^2)\ln n)/2$, we have $\phi_{k+1} > \phi_k$, and therefore for
  sufficiently large $n$, we have $r_{k+1} > r_k$.
\end{thm}

Finally, using these techniques, we can provide a tight characterization of the
probability that the Rooney Rule produces a positive change. 
Specifically, we find
the probability that the Rooney Rule has a positive effect conditioned on the
event that it changes the outcome.
\begin{thm} \label{thm:prob_pos}
  \[
    \Pr\b{\xan > \ynko | \xan < \beta \ynko} \apeq 1-\p{\frac{1 + \alpha
    \beta\tpd}{1+\alpha}}^k.
  \]
\end{thm}
To determine whether the Rooney Rule is more likely than not to produce
a positive effect (conditioned on changing the outcome), 
we can compare the right-hand side to $1/2$.

Note that in the case of infinite bias, the right-hand side
becomes $1-(1+\alpha)^{-k}$, and
thus, the Rooney Rule produces positive change with probability at least $1/2$
if and only if $\alpha \ge \sqrt[k]{2} - 1$. 
It is interesting to observe that this means with infinite bias, 
the condition is independent of $\delta$;
in contrast, when considering the effect on 
the expected value with infinite bias,
as we did in Section~\ref{sec:inf_bias}, the expected
change in utility due to the Rooney Rule did depend on $\delta$.

\subsection{Maximum Likelihood Estimation of $\beta$}
The techniques established thus far make it possible to answer
other related questions, including the following type of question
that we consider in this section:
``Given some historical data on past selections, can we estimate the bias
present in the data?'' For example, suppose that for the last $m$ years, a firm
has selected one candidate for each year $i$ out of a pool of $\alpha_i n_i$
$X$-candidates and $n_i$ $Y$-candidates. If all applicants are
assumed to come from the same underlying distribution, then it is easy to see
that the expected number of $X$-selections (in the absence of bias) should be
\[
  \sum_{i=1}^m \frac{\alpha_i}{1 + \alpha_i},
\]
regardless of what distribution the applicants come from. However, if there is
bias in the selection procedure, then this quantity now depends on the bias model
and parameters of the distribution. In particular, in our model, we can use
Theorem \ref{thm:prob} to get
\[
  \Prb{\xan < \beta \yn} \apeq \frac{1}{1 + \alpha \beta^{-(1+\delta)}}.
\]
This gives us the following approximation for the likelihood of the data $D =
(M_1, \dots, M_m)$ given $\beta$, where $M_i$ is $1$ if an $X$-candidate was
selected in year $i$ and $0$ otherwise:
\[
  \prod_{i=1}^m (1-M_i) \cdot \frac{1}{1 + \alpha_i
  \beta^{-(1+\delta)}} + M_i \cdot \frac{\alpha_i \beta\tpd}{1 + \alpha_i
  \beta^{-(1+\delta)}}.
\]
Taking logarithms, this is
\[
  \sum_{i : M_i = 1} \log(\alpha_i \beta\tpd) - \sum_{i=1}^m \log(1 + \alpha_i
  \beta\tpd),
\]
and maximizing this is equivalent to maximizing
\[
  \sum_{i : M_i = 1} \log(\beta\tpd) - \sum_{i=1}^m \log(1 + \alpha_i
  \beta\tpd) = N\log(\beta\tpd) - \sum_{i=1}^m \log(1 + \alpha_i
  \beta\tpd)
\]
where $N$ is the number of $X$-candidates selected. Taking the derivative
with respect to $\beta$, we get
\[
  -(1+\delta) N \beta^{-1} + (1+\delta) \sum_{i=1}^m \frac{\alpha_i
    \beta^{-(2+\delta)}}{1+\alpha_i
  \beta\tpd}.
\]
Setting this equal to 0 and canceling common terms, we have
\begin{align*}
  \sum_{i=1}^m \frac{1}{1+\alpha_i^{-1} \beta^{1+\delta}} &= N
\end{align*}
Since each $1/(1 + \alpha_i^{-1} \beta^{1+\delta})$ is strictly monotonically
decreasing in $\beta$, there is a unique $\hat \beta$ for which equality holds,
meaning that the likelihood is uniquely maximized by $\hat \beta$, up to the $1 \pm
O((\ln n)^2/n)$ approximation we made for $\Prb{\xan < \beta \yn}$. In the
special case where $\alpha_i = \alpha$ for $i = 1, \dots, m$, then the solution
is given by
\[
  \hat \beta = \p{\p{\frac{m}{N} - 1}\alpha}^{1/(1+\delta)}.
\]

\section{Biased Selection with Bounded Distributions} \label{sec:bounded}
In this section, we consider a model in which applicants come from a
distribution with bounded support. Qualitatively, one would expect different
results here from those with power law distributions because in
a model with bounded distributions, 
we expect that for large $n$, the top order statistics of any
distribution will concentrate around the maximum of that distribution. As a
result, when there is even a small amount of bias against one population, for
large $n$ the probability that \textit{any} of the samples with the highest
perceived quality come from that population goes to 0. This means that the
Rooney Rule has an effect with high
probability, and the effect is positive if the unconditional expectation of the
top $X$-candidate is larger than the unconditional expectation of the
$Y$-candidate that it replaces.

We focus on the case when $\alpha = 1$, meaning we have equal numbers of applicants
from both populations. We use the same order statistic notation as before. While
all of our previous results have modeled the bias as a
multiplicative factor $\beta$, we can in fact show that in the bounded
distribution setting, for any model of bias
$\tilde{X}_{(k:n)} = b(X_{(k:n)})$ such
that $b(x) < T$ for $x \ge 0$, where $T$ is strictly less than the maximum of the
distribution, the Rooney Rule increases expected utility. Unlike in the previous
section the following theorem and analysis are by no means a tight
characterization; instead, this is an existence proof that for bounded
distributions, there is always a large enough $n$ such that the Rooney Rule
improves utility in expectation. We prove our results for continuous
distributions with support $[0,1]$, but a simple scaling argument shows that
this extends to any continuous distribution with bounded nonnegative support --
specifically, we scale a distribution such that $\inf_{x : f(x) > 0} = 0$ and
$\sup_{x : f(x) > 0} = 1$.
\begin{thm} \label{thm:bounded}
  If $f$ is a continuous probability density function on $[0,1]$ such that
  $\sup_{x : f(x) > 0} = 1$ and $\tilde{X}_{(n:n)} = b(\xn)$ is never more than $T
  < 1$, then for
  large enough $n$,
  \[
    \E{\xn - \ynn | b(\xn) < \ynn} > 0.
  \]
\end{thm}
While we the defer the full proof to Appendix~\ref{app:missing_bd},
the strategy for the proof is as follows:
\begin{enumerate}
  \item With high probability, $\xn$ and $\ynn$ are both large.
  \item Whenever $\xn$ and $\ynn$ are large, $\xn$ is
    significantly larger than $\ynn$.
  \item The gain from switching from $\ynn$ to $\xn$ when $\xn$
    and $\ynn$ are both large outweighs the loss when at least one of them is
    not large.
\end{enumerate}

\section{Conclusion}
In this work we have presented a model for implicit bias in a
selection problem motivated by settings including hiring and admissions, 
and we analyzed the Rooney Rule, which
can improve the quality of the resulting choices.
For one of the most natural settings of the problem, when
candidates are drawn from a power-law distribution,
we found a tight characterization of the conditions
under which the Rooney Rule improves the quality of the outcome.
In the process, we identified a 
number of counter-intuitive effects at work,
which we believe may also help provide insight into how we can
reason about implicit bias.
Our techniques also provided a natural solution to an inference problem in 
which we estimate parameters of a biased decision-making process.
Finally, we performed a similar
type of analysis on general bounded distributions.

There are a number of further directions in which these issues could
be investigated. One intriguing direction is to consider the possible
connections to the theory of optimal delegation (see e.g.
\cite{alonso-delegation}).\footnote{We thank Ilya Segal for suggesting
this connection to us.} In the study of delegation, a {\em principal}
wants a task carried out, but this task can only be performed by an {\em
agent} who may have a utility function that is different from the
principal's. In an important family of these models, the principal's
only recourse is to impose a restriction on the set of possible
actions taken by the agent, creating a more constrained task for the
agent to perform, in a way that can potentially improve the quality of the
eventual outcome from the principal's perspective. Our analysis of the
Rooney Rule can be viewed as taking place from the point of view of a principal
who is trying to recruit $k$ candidates, 
but where the process must be delegated
to an agent whose utilities for $X$-candidates and
$Y$-candidates are different from the principal's,
and who is the only party able to evaluate these candidates' potentials. 
The Rooney Rule, requiring that the agent select at least
one $X$-candidate, is an example of a mechanism that the principal
could impose to restrict the agent's set of possible actions, 
potentially improving the quality of
the selected candidates as measured by the principal.
More generally, it is interesting to ask whether there are other
contexts where such a link between delegation and this type of biased 
selection provides insight.

Our framework also makes it possible to naturally explore extensions 
of the basic model.
First, the model can be generalized to include noisy observations, 
potentially with a different level of noise for each group.
It would also be interesting to analyze generalizations of the
Rooney Rule; for example, if we were to define the
{\em $\ell^{\rm th}$-order Rooney Rule} to be the requirement that
at least $\ell$ of $k$ finalists must be from an underrepresented group,
we could ask which $\ell$ produces the greatest increase in utility
for a given set of parameters.  
Finally, we could benefit from a deeper undestanding of 
the function $\phi$ that appears in our main theorems. For example, while we
showed in Theorem~\ref{thm:monotone} that $\phi$ is monotone in
$\beta$ for $k = 2$, Figure~\ref{fig:phi} shows that $\phi$ is clearly
not monotone in $\delta$.  A better understaning of the function
$\phi$ may lead to new insights into our model and into the
phenomena it seeks to capture.

\paragraph{Acknowledgements} We thank Eric Parsonnet for his invaluable
technical insights.  This work was supported in part by a Simons
Investigator Grant and an NSF Graduate Fellowship.

\bibliographystyle{plain}
\bibliography{refs}

\appendix
\section{Missing Proofs for Section~\ref{sec:pl}} \label{app:missing_pl}
\begin{proof}[Proof of Theorems~\ref{thm:r_2} and~\ref{thm:rk}]
  We can expand the statement in Theorem \ref{thm:cond_x_k} to
  \begin{align*}
    &\E{\xan \cdot \ind{\xan < \beta \ynko}} \\
    &\apeq \E{\xan} \b{1 -
    (1+c^{-1})^{-\delta/(1+\delta)} \sum_{j=0}^{k-1}
    \binom{j-\frac{1}{1+\delta}}{j} (1+c)^{-j}} \\
    &\apeq (\alpha n)^{1/(1+\delta)}\Gamma\p{\frac{\delta}{1+\delta}}
    \b{1 - (1+c^{-1})^{-\delta/(1+\delta)} \sum_{j=0}^{k-1}
    \binom{j-\frac{1}{1+\delta}}{j} (1+c)^{-j}} \tag{By
      Lemma~\ref{lem:os_approx}}
  \end{align*}
  This gives us a ratio
  \begin{align*}
    r_k(\alpha, \beta, \delta) &= \frac{\E{\xan \cdot \ind{\xan < \beta \ynko}}}{\E{\ynko \cdot
    \ind{\xan < \beta \ynko}}} \\
    &\apeq \frac{(\alpha n)^{1/(1+\delta)}\Gamma\p{\frac{\delta}{1+\delta}}
    \b{1 - (1+c^{-1})^{-\delta/(1+\delta)} \sum_{j=0}^{k-1}
    \binom{j-\frac{1}{1+\delta}}{j} (1+c)^{-j}}}{(1 + c)^{-(k - 1/(1+\delta))}
    \frac{\Gamma\p{k - \frac{1}{1+\delta}}}{\Gamma(k)} n^{1/(1+\delta)}}
    \tag{Using Theorem~\ref{thm:cond_y_k}} \\
    &= \frac{\alpha^{1/(1+\delta)} \Gamma(k)\Gamma\p{\frac{\delta}{1+\delta}}
    (1+c)^{k-1/(1+\delta)}}{\Gamma\p{k - \frac{1}{1+\delta}}}
    \b{1 - (c^{-1}(1+c))^{-\delta/(1+\delta)} \sum_{j=0}^{k-1}
    \binom{j-\frac{1}{1+\delta}}{j} (1 + c)^{-j}} \\
    &= \frac{\alpha^{1/(1+\delta)} c^{\delta/(1+\delta)}
    \Gamma(k)\Gamma\p{\frac{\delta}{1+\delta}} (1+c)^{k - 1}}{\Gamma\p{k -
    \frac{1}{1+\delta}}} \b{(1+c^{-1})^{\delta/(1+\delta)} -
    \sum_{j=0}^{k-1} \binom{j-\frac{1}{1+\delta}}{j} (1 + c)^{-j}} \\
    &= \frac{\alpha^{1/(1+\delta)} c^{\delta/(1+\delta)}
    (1+c)^{k - 1}}{\binom{k-1-\frac{1}{1+\delta}}{k-1}} \b{(1+c^{-1})^{\delta/(1+\delta)} -
    \sum_{j=0}^{k-1} \binom{j-\frac{1}{1+\delta}}{j} (1 + c)^{-j}}
  \end{align*}
\end{proof}

\begin{proof}[Proof of Theorem~\ref{thm:r_2_increasing}]
  Since the only influence of $\beta$ is through $c$ and $c$ is
  decreasing in $\beta$, it is sufficient to show that
  \begin{equation*}
    \frac{\alpha^{1/(1+\delta)}\b{1 -
    (1+c^{-1})^{-\delta/(1+\delta)} \b{1 + \frac{\delta}{1+\delta}
    (1+c)^{-1}}}}{\frac{\delta}{1+\delta} (1+c)^{-1-\delta/(1+\delta)}}
  \end{equation*}
  is decreasing in $c$. Ignoring constants, this is
  \begin{align*}
    &\propto c^{\delta/(1+\delta)} (1+c) \b{(1 + c^{-1})^{\delta/(1+\delta)}
    - 1 - \frac{\delta}{1+\delta} (1+c)^{-1}} \\
    &= (1+c)^{1 + \delta/(1+\delta)} - c^{\delta/(1+\delta)}(1+c) -
    \frac{\delta}{1+\delta} c^{\delta/(1+\delta)} \\
    &= (1+c)^{1+\delta/(1+\delta)} - c^{1+\delta/(1+\delta)} - \p{1 +
      \frac{\delta}{1+\delta}} c^{\delta/(1+\delta)}
  \end{align*}
  This has derivative
  \begin{align*}
    \frac{\delta}{dc} &(1+c)^{1+\delta/(1+\delta)} - c^{1+\delta/(1+\delta)} -
    \p{1 + \frac{\delta}{1+\delta}} c^{\delta/(1+\delta)} \\
    &= \p{1 + \frac{\delta}{1+\delta}}(1+c)^{\delta/(1+\delta)} - \p{1 +
    \frac{\delta}{1+\delta}} c^{\delta/(1+\delta)} +
    \p{\frac{\delta}{1+\delta}} \p{1+\frac{\delta}{1+\delta}} c^{-1/(1+\delta)},
  \end{align*}
  which is negative if and only if
  \begin{align*}
    &(1+c)^{\delta/(1+\delta)} < c^{\delta/(1+\delta)} + \frac{\delta}{1+\delta}
    c^{-1/(1+\delta)} \\
    \Longleftrightarrow &(1+c)^{\delta/(1+\delta)} c^{-\delta/(1+\delta)} < 1 +
    \frac{\delta}{1+\delta} c^{-1} \\
    \Longleftrightarrow &(1+c^{-1})^{\delta/(1+\delta)} < 1 +
    \frac{\delta}{1+\delta} c^{-1}.
  \end{align*}
  This is true by Lemma \ref{lem:first_expansion}, which proves the theorem.
\end{proof}

\begin{proof}[Proof of Theorem \ref{thm:rk_increasing}]
  By Theorem~\ref{thm:rk},
  \begin{align*}
    \phi_k(\alpha, \beta, \delta)
    = \frac{\alpha^{1/(1+\delta)} c^{\delta/(1+\delta)}
    \Gamma(k)\Gamma\p{\frac{\delta}{1+\delta}} (1+c)^{k - 1}}{\Gamma\p{k -
    \frac{1}{1+\delta}}} \b{(1+c^{-1})^{\delta/(1+\delta)} -
    \sum_{j=0}^{k-1} \binom{j-\frac{1}{1+\delta}}{j} (1 + c)^{-j}}
  \end{align*}
  We use the fact that for $a, b \in
  \mathbb{Z}$ and $s \in \mathbb{R}$
  \[
    \frac{\Gamma(s-a+1)}{\Gamma(s-b+1)} = (-1)^{b-a}
    \frac{\Gamma(b-s)}{\Gamma(a-s)}.
  \]
  If the summation went to $\infty$, it would be
  \begin{align*}
    \sum_{j=0}^\infty \binom{j - \frac{1}{1+\delta}}{j} (1+c)^{-j}
    &= \sum_{j=0}^\infty (1 + c)^{-j} \frac{\Gamma\p{j +
    \frac{\delta}{1+\delta}}}{\Gamma\p{\frac{\delta}{1+\delta}}\Gamma(j+1)} \\
    &= \sum_{j=0}^\infty (1 + c)^{-j} (-1)^j \frac{\Gamma\p{1 -
      \frac{\delta}{1+\delta}}}{\Gamma\p{-j + 1 +
        \frac{\delta}{1+\delta}}\Gamma(j+1)} \\
    &= \sum_{j=0}^\infty \binom{-\frac{\delta}{1+\delta}}{j} (-(1+c)^{-1})^j \\
    &= (1 - (1+c)^{-1})^{-\delta/(1+\delta)} \\
    &= (1+c^{-1})^{\delta/(1+\delta)}
  \end{align*}
  Therefore,
  \[
    \sum_{j=0}^{k-1} \binom{j - \frac{1}{1+\delta}}{j} (1+c)^{-j} =
    (1+c^{-1})^{\delta/(1+\delta)} - 
    \sum_{j=k}^\infty \binom{j - \frac{1}{1+\delta}}{j} (1+c)^{-j}.
  \]
  Plugging this in,
  \begin{align*}
    \phi_k(\alpha, \beta, \delta)
    &= \frac{\alpha^{1/(1+\delta)} c^{\delta/(1+\delta)}
    \Gamma(k)\Gamma\p{\frac{\delta}{1+\delta}} (1+c)^{k - 1}}{\Gamma\p{k -
    \frac{1}{1+\delta}}} \sum_{j=k}^\infty \binom{j -
    \frac{1}{1+\delta}}{j} (1+c)^{-j} \\
    &= \frac{\alpha^{1/(1+\delta)} c^{\delta/(1+\delta)}
    \Gamma(k)\Gamma\p{\frac{\delta}{1+\delta}}}{\Gamma\p{k -
    \frac{1}{1+\delta}}(1+c)} \sum_{j=0}^\infty \binom{j + k -
    \frac{1}{1+\delta}}{j + k} (1+c)^{-j}
  \end{align*}
  With this, we can take
  \begin{align*}
    \phi_{k+1}(\alpha, \beta, \delta) &- \phi_k(\alpha, \beta, \delta) \\
    &= \frac{\alpha^{1/(1+\delta)} c^{\delta/(1+\delta)}
    \Gamma\p{\frac{\delta}{1+\delta}}}{(1+c)}
    \left[\frac{\Gamma(k+1)}{\Gamma\p{k + \frac{\delta}{1+\delta}}}
      \sum_{j=0}^\infty \binom{j + k +1 -
    \frac{1}{1+\delta}}{j + k + 1} (1+c)^{-j} \right. \\
    &- \left. \frac{\Gamma(k)}{\Gamma\p{k -
      \frac{1}{1+\delta}}} \sum_{j=0}^\infty \binom{j + k -
    \frac{1}{1+\delta}}{j + k} (1+c)^{-j}\right] \\
    &= \frac{\alpha^{1/(1+\delta)} c^{\delta/(1+\delta)}
    \Gamma(k)\Gamma\p{\frac{\delta}{1+\delta}}}{\Gamma\p{k -
    \frac{1}{1+\delta}}(1+c)}
    \left[\frac{k}{k - \frac{1}{1+\delta}}
      \sum_{j=0}^\infty \binom{j + k +1 -
    \frac{1}{1+\delta}}{j + k + 1} (1+c)^{-j} \right. \\
    &- \left. \sum_{j=0}^\infty \binom{j + k -
    \frac{1}{1+\delta}}{j + k} (1+c)^{-j}\right] \\
    &= \frac{\alpha^{1/(1+\delta)} c^{\delta/(1+\delta)}
    \Gamma(k)\Gamma\p{\frac{\delta}{1+\delta}}}{\Gamma\p{k -
    \frac{1}{1+\delta}}(1+c)} \sum_{j=0}^\infty  (1+c)^{-j}
    \b{\frac{k}{k - \frac{1}{1+\delta}} \binom{j + k +1 - \frac{1}{1+\delta}}{j
    + k + 1} - \binom{j + k - \frac{1}{1+\delta}}{j + k}}
  \end{align*}
  Thus, to show that $\phi_{k+1} > \phi_k$, it is sufficient to show that for $j \ge
  0$,
  \begin{align*}
    \frac{k}{k - \frac{1}{1+\delta}} \binom{j + k +1 - \frac{1}{1+\delta}}{j
    + k + 1} - \binom{j + k - \frac{1}{1+\delta}}{j + k} &> 0 \\
    \frac{k}{k - \frac{1}{1+\delta}} \frac{\Gamma\p{j + k + 1 +
    \frac{\delta}{1+\delta}}}{\Gamma(j+k+2)\Gamma\p{\frac{\delta}{1+\delta}}}
    - \frac{\Gamma\p{j + k +
    \frac{\delta}{1+\delta}}}{\Gamma(j+k+1)\Gamma\p{\frac{\delta}{1+\delta}}}
    &> 0 \\
    \frac{k}{k - \frac{1}{1+\delta}} \frac{j+k+\frac{\delta}{1+\delta}}{j+k+1} -
    1 &> 0 \tag{$\Gamma(x+1) = x\Gamma(x)$} \\
    \frac{k-\frac{1}{1+\delta} + (j+1)}{k + (j+1)} &> \frac{k -
      \frac{1}{1+\delta}}{k}
  \end{align*}
  The last inequality holds by Lemma \ref{lem:frac_bigger}. As a result a
  result, $\phi_{k+1} > \phi_k$, proving the theorem.
\end{proof}
\begin{proof}[Proof Theorem~\ref{thm:prob_pos}]
  We want to find
  \[
    \Pr\b{\xan > \ynko | \xan < \beta \ynko},
  \]
  or equivalently,
  \[
    \Pr\b{\xan < \ynko | \xan < \beta \ynko}.
  \]
  This can be written as
  \begin{equation}
    \frac{\Pr\b{\xan < \ynko \cap \xan < \beta \ynko}}{\Pr\b{\xan < \beta \ynko}} =
    \frac{\Pr\b{\xan < \ynko}}{\Pr\b{\xan < \beta \ynko}}.
    \label{eq:ratio}
  \end{equation}
  By Theorem~\ref{thm:prob}, the numerator can be approximated by
  $(1+\alpha)^{-k}$ while the denominaotr is approximately $(1+\alpha
  \beta\tpd)^{-k}$. Thus, we have
  \begin{equation*}
    \Pr\b{\xan < \ynko | \xan < \beta \ynko} \approxident \frac{(1+\alpha
    \beta^{-(1+\delta)})^k}{(1+\alpha)^k} = \p{\frac{1+\alpha
    \beta^{-(1+\delta)}}{1+\alpha}}^k,
  \end{equation*}
  and therefore
  \[
    \Pr\b{\xan > \ynko | \xan < \beta \ynko} \approxident 1 - \p{\frac{1+\alpha
    \beta^{-(1+\delta)}}{1+\alpha}}^k.
  \]
\end{proof}

\section{Additional Theorems for Power Laws} \label{app:pl_proofs}
\begin{thm} \label{thm:cond_x_k}
  \[
    \E{\xan \cdot \ind{\xan < \beta \ynko}} \apeq \E{\xan} \b{1 -
    (1+c^{-1})^{-\delta/(1+\delta)} \sum_{j=0}^{k-1}
    \binom{j-\frac{1}{1+\delta}}{j} (1+c)^{-j}}
  \]
  where $c = \alpha \beta^{-(1+\delta)}$.
\end{thm}
\begin{proof}
  First, observe that
  \[
    \E{\xan \cdot \ind{\xan < \beta \ynko}} = \E{\xan} - \E{\xan \cdot
    \ind{\xan \ge \beta \ynko}}.
  \]
  Next, we use the fact that
  \[
    \E{\xan \cdot \ind{\xan \ge \beta \ynko}} = \int_\beta^{\infty} x
    \pdf{\alpha n}{\alpha n}(x) \cdf{n-k+1}{n}\p{\frac{x}{\beta}} \dx.
  \]
  We know that
  \begin{align*}
    \int_\beta^{\infty} x \pdf{\alpha n}{\alpha n}(x)
    \cdf{n-k+1}{n}\p{\frac{x}{\beta}} \dx &=
    \int_\beta^{\p{\frac{\alpha n}{\ln n}}^{1/(1+\delta)}} x \pdf{\alpha
    n}{\alpha n}(x) \cdf{n-k+1}{n}\p{\frac{x}{\beta}} \dx \\
    &+ \int_{\p{\frac{\alpha n}{\ln n}}^{1/(1+\delta)}}^{\infty} x \pdf{\alpha
    n}{\alpha n}(x) \cdf{n-k+1}{n}\p{\frac{x}{\beta}} \dx, \numberthis
    \label{eq:x_split}
  \end{align*}
  and
  \begin{align*}
    \int_\beta^{\p{\frac{\alpha n}{\ln n}}^{1/(1+\delta)}} x \pdf{\alpha
    n}{\alpha n}(x) \cdf{n-k+1}{n}\p{\frac{x}{\beta}} \dx
    &\le \p{\frac{\alpha n}{\ln n}}^{1/(1+\delta)} \cdf{\alpha n}{\alpha n}
    \p{\p{\frac{\alpha n}{\ln n}}^{1/(1+\delta)}} \\
    &\le \p{\frac{\alpha n}{\ln n}}^{1/(1+\delta)} \cdot \frac{1}{n}
  \end{align*}
  by Lemma \ref{lem:cdf_bound}.
  The second term of \eqref{eq:x_split} is
  \begin{align*}
    &(1+\delta) \alpha n \int_{\p{\frac{\alpha n}{\ln n}}^{1/(1+\delta)}}^{\infty}
    (1 - x^{-(1+\delta)})^{\alpha n - 1} x^{-(1+\delta)} \sum_{j=0}^{k-1}
    \binom{n}{j} \p{1 - \p{\frac{x}{\beta}}^{-(1+\delta)}}^{n-j}
    \p{\frac{x}{\beta}}^{-j(1+\delta)} \dx \\
    &=(1+\delta) \alpha n \sum_{j=0}^{k-1} \binom{n}{j} \beta^{j(1+\delta)}
    \int_{\p{\frac{\alpha n}{\ln n}}^{1/(1+\delta)}}^{\infty} (1 -
    x^{-(1+\delta)})^{\alpha n - 1} \p{x^{-(1+\delta)}}^{j+1} \p{1 -
    \p{\frac{x}{\beta}}^{-(1+\delta)}}^{n-j} \dx
  \end{align*}
  Next, we show that for $x \ge \p{\frac{\alpha n}{\ln n}}^{1/(1+\delta)}$,
  \[
    \p{1 - \p{\frac{x}{\beta}}^{-(1+\delta)}}^{n-j}
    \apeq (1 - x^{-(1+\delta)})^{\beta^{1+\delta}n - j}.
  \]
  We begin with
  \[
    \p{1 - \p{\frac{x}{\beta}}^{-(1+\delta)}}^{n-j}
    \apeq (1 - x^{-(1+\delta)})^{\beta^{1+\delta}(n-j)} = 
    (1 - x^{-(1+\delta)})^{\beta^{1+\delta}n - j}
    (1 - x^{-(1+\delta)})^{-j(\beta^{1+\delta}-1)}.
  \]
  Note that $(1 - x^{-(1+\delta)})^{-j(\beta^{1+\delta}-1)} \ge 1$, and by Lemma
  \ref{lem:taylor_1_yz},
  \begin{align*}
    (1 - x^{-(1+\delta)})^{-j(\beta^{1+\delta}-1)} = 1 + j(\beta^{1+\delta}-1)
    x^{-(1+\delta)} + O\p{\frac{1}{n}} \apeq 1.
  \end{align*}
  because $j \le \ln n$. Thus, $(1 - x^{-(1+\delta)})^{\beta^{1+\delta}n - j} (1
  - x^{-(1+\delta)})^{-j(\beta^{1+\delta}-1)} \apeq (1 -
  x^{-(1+\delta)})^{\beta^{1+\delta}n - j}$.
  Therefore, this becomes
  \[
    (1+\delta) \alpha n \sum_{j=0}^{k-1} \binom{n}{j} \beta^{j(1+\delta)}
    \int_{\p{\frac{\alpha n}{\ln n}}^{1/(1+\delta)}}^{\infty} \p{1 -
      x^{-(1+\delta)}}^{\beta^{1+\delta} n(1+c) - j - 1}
    \p{x^{-(1+\delta)}}^{j+1} \dx.
  \]
  We'll now try to relate the $j$th term in this summation to the order
  statistic $\znjo$. We know that
  \begin{align*}
    &\E{\znjo} \\
    &= \int_1^\infty z \pdf{\beta^{1+\delta} n(1+c)-j}{\beta^{1+\delta}n(1+c)}(z) \dz \\
    &= (1+\delta) (j+1) \binom{\beta^{1+\delta} n(1+c)}{j+1} \int_1^\infty \p{1 -
      z^{-(1+\delta)}}^{\beta^{1+\delta} n(1+c)-j-1}
      \p{z^{-(1+\delta)}}^{j+1} \dz.
  \end{align*}
  Using this, we have
  \begin{align*}
    \int_{\p{\frac{\alpha n}{\ln n}}^{1/(1+\delta)}}^\infty x \pdf{\alpha
    n}{\alpha n}(x)
    \cdf{n-k+1}{n}\p{\frac{x}{\beta}} \dx
    &\apeq \sum_{j=0}^{k-1} \frac{\alpha n \beta^{j(1+\delta)}
    \binom{n}{j}}{(j+1)\binom{\beta^{1+\delta}n(1+c)}{j+1}} \left[\E{\znjo}
      \right. \\
    &- \left.\int_1^{\p{\frac{\alpha n}{\ln n}}^{1/(1+\delta)}} z
  \pdf{\beta^{1+\delta} n(1+c)-j}{\beta^{1+\delta}n(1+c)}(z) \dz\right]
  \end{align*}
  We'll show that this last multiplicative term is approximately $\E{\znjo}$.
  Observe that
  \begin{align*}
    \int_1^{\p{\frac{\alpha n}{\ln n}}^{1/(1+\delta)}} &z
    \pdf{\beta^{1+\delta} n(1+c)-j}{\beta^{1+\delta}n(1+c)}(z) \dz \\
    &\le \p{\frac{\alpha n}{\ln n}}^{1/(1+\delta)} \int_1^{\p{\frac{\alpha n}{\ln
    n}}^{1/(1+\delta)}} \pdf{\beta^{1+\delta} n(1+c)-j}{\beta^{1+\delta}n(1+c)}(z) \dz \\
    &= \p{\frac{\alpha n}{\ln n}}^{1/(1+\delta)}
    \cdf{\beta^{1+\delta}n(1+c)-j}{\beta^{1+\delta}n(1+c)}\p{\p{\frac{\alpha
    n}{\ln n}}^{1/(1+\delta)}} \\
    &\le \p{\frac{\alpha n}{\ln n}}^{1/(1+\delta)} \frac{\sqrt{k}}{n}
  \end{align*}
  by Lemma \ref{lem:x_cdf_bound}. This means
  \begin{align*}
    \sum_{j=0}^{k-1} &\frac{\alpha n \beta^{j(1+\delta)}
    \binom{n}{j}}{(j+1)\binom{\beta^{1+\delta}n(1+c)}{j+1}} \E{\znjo} \\
    &\ge \sum_{j=0}^{k-1} \frac{\alpha n \beta^{j(1+\delta)}
    \binom{n}{j}}{(j+1)\binom{\beta^{1+\delta}n(1+c)}{j+1}} \b{\E{\znjo}
    - \p{\frac{\alpha n}{\ln n}}^{1/(1+\delta)} \frac{\sqrt{k}}{n}} \\
    &\apeq \sum_{j=0}^{k-1} \frac{\alpha n \beta^{j(1+\delta)}
    \binom{n}{j}}{(j+1)\binom{\beta^{1+\delta}n(1+c)}{j+1}} \E{\znjo}
    \tag{by Lemma \ref{lem:zlnn}}
  \end{align*}
  Next, we deal with the $n\beta^{j(1+\delta)}
  \binom{n}{j}/((j+1)\binom{\beta^{1+\delta}n(1+c)}{j+1})$ terms. These are
  \begin{equation}
    \frac{n\beta^{j(1+\delta)}
    \binom{n}{j}}{(j+1)\binom{\beta^{1+\delta}n(1+c)}{j+1}} = \frac{n(n-1)
    \cdots (n-j+1)}{\beta^{1+\delta}n(1+c) (\beta^{1+\delta}n(1+c)-1) \cdots
    (\beta^{1+\delta}n(1+c) - j+1)} \cdot \frac{n
      \beta^{j(1+\delta)}}{\beta^{1+\delta} n(1+c)-j}.
      \label{eq:x_binom_ratio}
  \end{equation}
  Each term $(n-\ell)/(\beta^{1+\delta}n(1+c)-\ell)$ is between
  $1/(\beta^{1+\delta}(1+c))$ and $1/(\beta^{1+\delta}(1+c)) \cdot
  (1-\ell/n)$.
  This means
  \begin{align*}
    \frac{1}{(\beta^{1+\delta}(1+c))^j} \ge \prod_{\ell=0}^j
    \frac{n-\ell}{\beta^{1+\delta}n(1+c)-\ell} \ge
    \prod_{\ell=0}^{j}
    \frac{1}{\beta^{1+\delta}(1+c)} \p{1-\frac{\ell}{n}}
    &\ge
    \p{1-\frac{j^2}{n}}
    \apeq \frac{1}{(\beta^{1+\delta}(1+c))^j}
  \end{align*}
  since $j \le k \le ((1-c^2)/2) \ln n$.
  Multiplying by the second term in \eqref{eq:x_binom_ratio}, which is
  \[
    \frac{n \beta^{j(1+\delta)}}{\beta^{1+\delta} n(1+c)-j} \apeq
    \frac{\beta^{(j-1)(1+\delta)}}{1+c},
  \]
  we have
  \[
    \frac{n\beta^{j(1+\delta)}
    \binom{n}{j}}{(j+1)\binom{\beta^{1+\delta}n(1+c)}{j+1}}
    \apeq \frac{1}{\beta^{1+\delta}(1+c)^{j+1}}.
  \]
  As a result,
  \begin{align*}
    \int_{\p{\frac{\alpha n}{\ln n}}^{1/(1+\delta)}}^\infty x \pdf{\alpha
    n}{\alpha n}(x)
    \cdf{n-k+1}{n}\p{\frac{x}{\beta}} \dx
    &\apeq \sum_{j=0}^{k-1} \frac{\alpha}{\beta^{1+\delta}(1+c)^{j+1}}
    \E{\znjo} \\
    &= \sum_{j=0}^{k-1} \frac{c}{(1+c)^{j+1}} \E{\znjo} \numberthis
    \label{eq:x_sum_z}
  \end{align*}
  Finally, note that
  \begin{align*}
    \E{\znjo}
    &= \E{\znbc} \frac{\Gamma(j + \delta/(1+\delta))}{\Gamma(\delta/(1+\delta))
    \Gamma(j+1)} \\
    &\apeq (\beta^{1+\delta} n(1+c))^{1/(1+\delta)} \frac{\Gamma(j +
    \delta/(1+\delta))}{\Gamma(j+1)} \\
    &= \beta(1+c)^{1/(1+\delta)} n^{1/(1+\delta)} \frac{\Gamma(j +
    \delta/(1+\delta))}{\Gamma(j+1)} \\
    &= \frac{\beta}{\alpha^{1/(1+\delta)}}(1+c)^{1/(1+\delta)} (\alpha
    n)^{1/(1+\delta)} \frac{\Gamma(j + \delta/(1+\delta))}{\Gamma(j+1)} \\
    &\apeq c^{-1/(1+\delta)} (1+c)^{1/(1+\delta)} \E{\xan} \frac{\Gamma(j +
    \delta/(1+\delta))}{\Gamma(\delta/(1+\delta)) \Gamma(j+1)}
  \end{align*}
  Substituting back to \eqref{eq:x_sum_z},
  \begin{align*}
    \int_{\p{\frac{\alpha n}{\ln n}}^{1/(1+\delta)}}^\infty x \pdf{\alpha
    n}{\alpha n}(x)
    \apeq \E{\xan} c^{\delta/(1+\delta)} \sum_{j=0}^{k-1} (1+c)^{-(j +
    \delta/(1+\delta))} \frac{\Gamma(j +
    \delta/(1+\delta))}{\Gamma(\delta/(1+\delta)) \Gamma(j+1)}
  \end{align*}
  Going back to \eqref{eq:x_split},
  \begin{align*}
    \E{\xan \cdot \ind{\xan > \beta \ynko}}
    &\apeq \E{\xan} c^{\delta/(1+\delta)} \sum_{j=0}^{k-1} (1+c)^{-(j +
    \delta/(1+\delta))} \frac{\Gamma(j +
    \delta/(1+\delta))}{\Gamma(\delta/(1+\delta)) \Gamma(j+1)} \\
    &+ \int_\beta^{\p{\frac{\alpha n}{\ln n}}^{1/(1+\delta)}} x \pdf{\alpha
    n}{\alpha n}(x) \cdf{n-k+1}{n}\p{\frac{x}{\beta}} \dx \\
    &\le \E{\xan} c^{\delta/(1+\delta)} \sum_{j=0}^{k-1} (1+c)^{-(j +
    \delta/(1+\delta))} \frac{\Gamma(j +
    \delta/(1+\delta))}{\Gamma(\delta/(1+\delta)) \Gamma(j+1)} \\
    &+ \p{\frac{\alpha n}{\ln n}}^{1/(1+\delta)} \frac{(\ln n)^2}{n} \\
    &\apeq \E{\xan} c^{\delta/(1+\delta)} \sum_{j=0}^{k-1} (1+c)^{-(j +
    \delta/(1+\delta))} \frac{\Gamma(j +
    \delta/(1+\delta))}{\Gamma(\delta/(1+\delta)) \Gamma(j+1)} \\
  \end{align*}
  Therefore,
  \[
    \E{\xan \cdot \ind{\xan < \beta \ynko}} \apeq \E{\xan} \b{1 - c^{\delta/(1+\delta)}
    \sum_{j=0}^{k-1} (1+c)^{-(j + \delta/(1+\delta))} \frac{\Gamma\p{j +
      \frac{\delta}{1+\delta}}}{\Gamma\p{\frac{\delta}{1+\delta}} \Gamma(j+1)}}.
  \]
  We can simplify this to
  \[
    \E{\xan \cdot \ind{\xan < \beta \ynko}} \apeq \E{\xan} \b{1 -
    (1+c^{-1})^{-\delta/(1+\delta)} \sum_{j=0}^{k-1} (1+c)^{-j} \frac{\Gamma\p{j +
    \frac{\delta}{1+\delta}}}{\Gamma\p{\frac{\delta}{1+\delta}}\Gamma(j+1)}}.
  \]
  Using the definition
  \[
    \binom{a}{b} = \frac{\Gamma(a+1)}{\Gamma(b+1) \Gamma(a-b+1)},
  \]
  this is
  \[
    \E{\xan \cdot \ind{\xan < \beta \ynko}} \apeq \E{\xan} \b{1 -
    (1+c^{-1})^{-\delta/(1+\delta)} \sum_{j=0}^{k-1}
    \binom{j-\frac{1}{1+\delta}}{j} (1+c)^{-j}}.
  \]
\end{proof}
\begin{thm} \label{thm:cond_y_k}
  \[
    \E{\ynko \cdot \ind{\xan < \beta \ynko}} \apeq (1 + \alpha
    \beta^{-(1+\delta)})^{-(k - 1/(1+\delta))} \E{\ynko}
  \]
\end{thm}
\begin{proof}
  We begin with
  \[
    \E{\ynko \cdot \ind{\xan < \beta \ynko}}
    = \int_1^\infty y \pdf{n-k+1}{n}(y) \cdf{\alpha n}{\alpha n}(\beta y) \dy.
  \]
  Let $c = \alpha \beta^{-(1+\delta)}$. Break this up into
  \begin{align*}
    \int_1^\infty y \pdf{n-k+1}{n}(y) \cdf{\alpha n}{\alpha n}(\beta y) \dy
    &= \int_1^{\p{\frac{cn}{\ln n}}^{1/(1+\delta)}} y \pdf{n-k+1}{n}(y)
    \cdf{\alpha n}{\alpha n}(\beta y) \dy \\
    &+ \int_{\p{\frac{cn}{\ln n}}^{1/(1+\delta)}}^\infty y \pdf{n-k+1}{n}(y)
    \cdf{\alpha n}{\alpha n}(\beta y) \dy.
    \numberthis \label{eq:split_ynko}
  \end{align*}
  The first term is
  \begin{align*}
    \int_1^{\p{\frac{cn}{\ln n}}^{1/(1+\delta)}} &y \pdf{n-k+1}{n}(y)
    \cdf{\alpha n}{\alpha n}(\beta y) \dy \\
    &\le \cdf{\alpha n}{\alpha n}\p{\beta \p{\frac{cn}{\ln n}}^{1/(1+\delta)}}
    \int_1^{\p{\frac{cn}{\ln n}}^{1/(1+\delta)}} y \pdf{n-k+1}{n}(y) \dy \\
    &\le \cdf{\alpha n}{\alpha n}\p{\beta \p{\frac{cn}{\ln n}}^{1/(1+\delta)}}
    \E{\ynko} \\
    &\le \frac{\E{\ynko}}{n}
  \end{align*}
  by Lemma \ref{lem:cdf_bound}.

  For the second term in \eqref{eq:split_ynko}, we have
  \begin{align*}
    \int_{\p{\frac{cn}{\ln n}}^{1/(1+\delta)}}^\infty &y \pdf{n-k+1}{n}(y)
    \cdf{\alpha n}{\alpha n}(\beta y) \dy \\
    &= (1+\delta) k \binom{n}{k} \int_{\p{\frac{cn}{\ln
    n}}^{1/(1+\delta)}}^\infty \p{1 - y^{-(1+\delta)}}^{n-k}
    \p{y^{-(1+\delta)}}^k \p{1 - \p{\beta y}^{-(1+\delta)}}^{\alpha n} \dy
  \end{align*}
  By Lemma \ref{lem:y_approx_c}, for all $y \ge (cn/\ln n)^{1/(1+\delta)}$,
  \[
    \p{1-(\beta y)^{-(1+\delta)}}^{\alpha n} \apeq \p{1-y^{-(1+\delta)}}^{cn}.
  \]
  Therefore,
  \begin{align*}
    \int_{\p{\frac{cn}{\ln n}}^{1/(1+\delta)}}^\infty &y \pdf{n-k+1}{n}(y)
    \cdf{\alpha n}{\alpha n}(\beta y) \dy \\
    &\apeq (1+\delta) k \binom{n}{k} \int_{\p{\frac{cn}{\ln
    n}}^{1/(1+\delta)}}^\infty \p{1 - y^{-(1+\delta)}}^{n-k + c n}
    \p{y^{-(1+\delta)}}^k \dy \\
    &= (1+\delta) k \binom{n}{k} \int_{\p{\frac{cn}{\ln
    n}}^{1/(1+\delta)}}^\infty \p{1 - y^{-(1+\delta)}}^{n(1+c)-k}
    \p{y^{-(1+\delta)}}^k \dy.
  \end{align*}
  We'll now try to relate this to the order statistic $\znko$. We know that
  \begin{align*}
    \E{\znko}
    &= \int_1^\infty z \pdf{n(1+c)-k+1}{n(1+c)}(z) \dz \\
    &= (1+\delta) k \binom{n(1+c)}{k} \int_1^\infty \p{1 -
      z^{-(1+\delta)}}^{n(1+c)-k} \p{z^{-(1+\delta)}}^k \dz.
  \end{align*}
  Using this, we have
  \begin{align*}
    \int_{\p{\frac{cn}{\ln n}}^{1/(1+\delta)}}^\infty &y \pdf{n-k+1}{n}(y)
    \cdf{\alpha n}{\alpha n}(\beta y) \dy \\
    &\apeq \frac{\binom{n}{k}}{\binom{n(1+c)}{k}} \b{\E{\znko} -
    \int_1^{\p{\frac{cn}{\ln n}}^{1/(1+\delta)}} y
    \pdf{n(1+c)-k+1}{n(1+c)}(y) \dy}. \numberthis \label{eq:higher_os}
  \end{align*}
  From here, we'll show that the term being subtracted is only a
  $\frac{\sqrt{\ln n}}{n}$ fraction
  of $\E{\znko}$. To do so, note that
  \begin{align*}
    \int_1^{\p{\frac{cn}{\ln n}}^{1/(1+\delta)}} y
    \pdf{n(1+c)-k+1}{n(1+c)}(y) \dy
    &\le \p{\frac{cn}{\ln n}}^{1/(1+\delta)} \int_1^{\p{\frac{cn}{\ln
    n}}^{1/(1+\delta)}} \pdf{n(1+c)-k+1}{n(1+c)}(y) \dy \\
    &= \p{\frac{cn}{\ln n}}^{1/(1+\delta)}
    \cdf{n(1+c)-k+1}{n(1+c)}\p{\p{\frac{cn}{\ln n}}^{1/(1+\delta)}} \\
    &\le \p{\frac{cn}{\ln n}}^{1/(1+\delta)} \p{\frac{\sqrt{k}}{n}}
  \end{align*}
  By Lemma \ref{lem:y_cdf_bound}. Lemma \ref{lem:zlnn} gives us
  \begin{align*}
    \E{\znko} &\ge \E{\znko} - \int_1^{\p{\frac{cn}{\ln n}}^{1/(1+\delta)}} y
    \pdf{n(1+c)-k+1}{n(1+c)}(y) \dy \\
    &\ge \E{\znko} - \p{\frac{cn}{\ln n}}^{1/(1+\delta)} \p{\frac{\sqrt{k}}{n}}
    \\
    &\ge \E{\znko} \p{1 - \frac{\sqrt{k}}{n}} \\
    &\apeq \E{\znko}
  \end{align*}
  Combining with \eqref{eq:higher_os}, Lemma \ref{lem:squeeze} yields
  \begin{align*}
    \int_{\p{\frac{cn}{\ln n}}^{1/(1+\delta)}}^\infty y \pdf{n-k+1}{n}(y)
    \cdf{\alpha n}{\alpha n}(\beta y) \dy \apeq
    \frac{\binom{n}{k}}{\binom{n(1+c)}{k}} \E{\znko} \numberthis
    \label{eq:higher_os2}
  \end{align*}
  By Lemma~\ref{lem:binom_ratio},
  \[
    \frac{\binom{n}{k}}{\binom{n(1+c)}{k}} \apeq \frac{1}{(1+c)^k}.
  \]
  Putting this into \eqref{eq:higher_os2},
  \begin{align*}
    \int_{\p{\frac{cn}{\ln n}}^{1/(1+\delta)}}^\infty y \pdf{n-k+1}{n}(y)
    \cdf{\alpha n}{\alpha n}(\beta y) \dy \apeq \frac{1}{(1+c)^k} \E{\znko}.
  \end{align*}
  Finally, note that
  \begin{align*}
    \E{\znko}
    &= \E{\znc} \frac{\Gamma(k - 1/(1+\delta))}{\Gamma(\delta/(1+\delta))
    \Gamma(k)} \\
    &\apeq (n(1+c))^{1/(1+\delta)} \frac{\Gamma(k -
    1/(1+\delta))}{\Gamma(k)} \\
    &= (1+c)^{1/(1+\delta)} n^{1/(1+\delta)} \frac{\Gamma(k -
    1/(1+\delta))}{\Gamma(k)} \\
    &\apeq (1+c)^{1/(1+\delta)} \E{\yn} \frac{\Gamma(k -
    1/(1+\delta))}{\Gamma(\delta/(1+\delta)) \Gamma(k)} \\
    &= (1+c)^{1/(1+\delta)} \E{\ynko} 
  \end{align*}
  Substituting into \eqref{eq:split_ynko},
  \begin{align*}
    \E{\ynko \cdot \ind{\xan \le \beta \ynko}}
    &\apeq \E{\ynko} \p{(1+c)^{-(k-1/(1+\delta))} + \frac{1}{n}} \\
    &\apeq \E{\ynko} (1+\alpha \beta^{-(1+\delta)})^{-(k-1/(1+\delta))}
  \end{align*}
  since $c = a\beta^{-1(1+\delta)}$, proving the theorem.
\end{proof}
\begin{thm}
  \label{thm:prob}
  \[
    \Pr\b{\xan < \beta \ynko} \apeq (1+c)^{-k}.
  \]
\end{thm}
\begin{proof}
  Begin with
  \begin{align*}
    \Pr\b{\xan < \beta \ynko} &= \int_1^\infty \pdf{n-k+1}{n}(y) \cdf{\alpha
    n}{\alpha n}(\beta y) \dy \\
    &= \int_1^{\p{\frac{cn}{\ln n}}^{1/(1+\delta)}} \pdf{n-k+1}{n}(y) \cdf{\alpha
    n}{\alpha n}(\beta y) \dy \\
    &+ \int_{\p{\frac{cn}{\ln n}}^{1/(1+\delta)}}^\infty \pdf{n-k+1}{n}(y)
    \cdf{\alpha n}{\alpha n}(\beta y) \dy \numberthis \label{eq:prob_split}
  \end{align*}
  Observe that
  \begin{align*}
    \int_1^{\p{\frac{cn}{\ln n}}^{1/(1+\delta)}} \pdf{n-k+1}{n}(y) \cdf{\alpha
    n}{\alpha n}(\beta y) \dy
    &\le \cdf{\alpha n}{\alpha n} \p{\beta \p{\frac{cn}{\ln n}}^{1/(1+\delta)}}
    \cdf{n-k+1}{n} \p{\p{\frac{cn}{\ln n}}^{1/(1+\delta)}}\\
    &\le \cdf{\alpha n}{\alpha n} \p{\beta \p{\frac{cn}{\ln n}}^{1/(1+\delta)}}
    \\
    &\le \p{1 - \beta^{-(1+\delta)} \p{\frac{\ln n}{cn}}}^{\alpha n} \\
    &\le \exp \p{-\alpha \beta^{-(1+\delta)}\frac{\ln n}{cn}} \\
    &= \frac{1}{n}
  \end{align*}
  Next, we have
  \begin{align*}
    \int_{\p{\frac{cn}{\ln n}}^{1/(1+\delta)}}^\infty \pdf{n-k+1}{n}(y)
    \cdf{\alpha n}{\alpha n}(\beta y) \dy
    &= \int_{\p{\frac{cn}{\ln n}}^{1/(1+\delta)}}^\infty (1-(\beta
    y)^{-(1+\delta)})^{\alpha n} \pdf{n-k+1}{n}(y) \dy
  \end{align*}
  By Lemma \ref{lem:y_approx_c}, for $y \ge (cn/\ln n)^{1/(1+\delta)}$,
  \[
    \p{1-(\beta y)^{-(1+\delta)}}^{\alpha n} \apeq \p{1-y^{-(1+\delta)}}^{cn},
  \]
  so
  \begin{align*}
    \int_{\p{\frac{cn}{\ln n}}^{1/(1+\delta)}}^\infty &\pdf{n-k+1}{n}(y)
    \cdf{\alpha n}{\alpha n}(\beta y) \dy \\
    &\apeq \int_{\p{\frac{cn}{\ln n}}^{1/(1+\delta)}}^\infty
    (1-y^{-(1+\delta)})^{c n} \pdf{n-k+1}{n}(y) \dy \\
    &= (1+\delta) k \binom{n}{k} 
    \int_{\p{\frac{cn}{\ln n}}^{1/(1+\delta)}}^\infty
    (1-y^{-(1+\delta)})^{n(1+c)-k} (y^{-(1+\delta)})^k y^{-1} \dy \\
    &= \frac{\binom{n}{k}}{\binom{n(1+c)}{k}}
    \int_{\p{\frac{cn}{\ln n}}^{1/(1+\delta)}}^\infty
    \pdf{n(1+c)-k+1}{n(1+c)} \dy
  \end{align*}
  From Lemma \ref{lem:y_cdf_bound}, we have
  \[
    \cdf{n(1+c)-k+1}{n(1+c)}\p{\p{\frac{cn}{\ln n}^{1/(1+\delta)}}} \le
    \frac{\sqrt{k}}{n},
  \]
  so
  \begin{align*}
    \int_{\p{\frac{cn}{\ln n}}^{1/(1+\delta)}}^\infty
    \pdf{n(1+c)-k+1}{n(1+c)} \dy
    &= 1 - \cdf{n(1+c)-k+1}{n(1+c)}\p{\p{\frac{cn}{\ln n}^{1/(1+\delta)}}} \\
    &\ge 1 - \frac{\sqrt{k}}{n} \\
    &\approx 1.
  \end{align*}
  Therefore,
  \begin{align*}
    \int_{\p{\frac{cn}{\ln n}}^{1/(1+\delta)}}^\infty \pdf{n-k+1}{n}(y)
    \cdf{\alpha n}{\alpha n}(\beta y) \dy \apeq
    \frac{\binom{n}{k}}{\binom{n(1+c)}{k}} \apeq \frac{1}{(1+c)^k}
  \end{align*}
  by Lemma~\ref{lem:binom_ratio}. By \eqref{eq:prob_split}, this means
  \[
    \Pr\b{\xan < \beta \ynko} \apeq (1+c)^{-k}.
  \]
\end{proof}

\section{Lemmas for the Equivalence Definition} \label{app:equiv}
\begin{lemma}[Transitivity]
  If $f(n) \apeqq{a_1}{n_1} g(n)$ and $g(n) \apeqq{a_2}{n_2} h(n)$, then $f(n)
  \apeq h(n)$.
  \label{lem:trans}
\end{lemma}
\begin{proof}
  \begin{align*}
    \frac{f(n)}{h(n)} &= \frac{f(n)}{g(n)} \cdot \frac{g(n)}{h(n)} \\
    &\le \p{1 + \frac{a_1 (\ln n)^2}{n}} \p{1 + \frac{a_2 (\ln n)^2}{n}} \\
    &\le 1 + \frac{(a_1 + a_2)(\ln n)^2}{n} + \frac{a_1 a_2 (\ln  n)^4}{n^2} \\
    &\le 1 + \frac{(a_1 + a_2 + a_1a_2)(\ln n)^2}{n}
  \end{align*}
  for all $n \ge \max(n_1, n_2)$, since $n \ge (\ln n)^2$. A symmetric argument
  holds for $h(n)/f(n)$. Thus, $f(n) \apeqq{a_1+a_2+a_1a_2}{\max(n_1,n_2)}
  h(n)$.
\end{proof}

\begin{lemma}[Linearity]
  If $f_1(n) \apeqq{a_1}{n_1} g_1(n)$ and $f_1(n) \apeqq{a_2}{n_2} g_2(n)$,
  then $b
  f_1(n) + c f_2(n) \apeq b g_1(n) + cg_2(n)$.
  \label{lem:linear}
\end{lemma}
\begin{proof}
  By Lemma \ref{lem:ratio_bound},
  \[
    \frac{bf_1(n) + cf_2(n)}{bg_1(n) + cg_2(n)} \le
    \max\p{\frac{f_1(n)}{g_1(n)}, \frac{f_2(n)}{g_2(n)}} \le
    \frac{\max(a_1, a_2)(\ln n)^2}{n}
  \]
  for $n \ge \max(n_1, n_2)$. A symmetric argument holds for the reciprocal.
  Therefore,
  \[
    b f_1(n) + c f_2(n) \apeqq{\max(a_1, a_2)}{\max(n_1, n_2)} b
    g_1(n) + cg_2(n).
  \]
\end{proof}

\begin{lemma}[Integrals]
  If $f(x, n) \apeqq{a}{n_0} g(x, n)$, then
  \[
    \int f(x, n) \dx \apeq \int g(x, n) \dx
  \]
  \label{lem:integrals}
\end{lemma}
\begin{proof}
  \[
    \frac{\int f(x, n) \dx}{\int g(x, n) \dx} = 
    \frac{\int g(x, n) \frac{f(x, n)}{g(x, n)} \dx}{\int g(x, n) \dx} \le
    \frac{\int g(x, n) \p{1 + \frac{a(\ln n)^2}{n}} \dx}{\int g(x, n) \dx} \le
    1 + \frac{a(\ln n)^2}{n}
  \]
  for $n \ge n_0$. A symmetric argument holds for the repciprocal, proving the
  lemma.
\end{proof}

\begin{lemma}
  If $f_1(n) \apeqq{a_1}{n_1} g_1(n)$ and $f_2(n) \apeqq{a_2}{n_2} g_2(n)$, then
  \[
    f_1(n) f_2(n) \apeq g_1(n) g_2(n).
  \]
  \label{lem:mult}
\end{lemma}
\begin{proof}
  \begin{align*}
    \frac{f_1(n) f_2(n)}{g_1(n) g_2(n)} &= \frac{f_1(n)}{g_1(n)} \cdot
    \frac{f_2(n)}{g_2(n)} \\
    &\le \p{1 + \frac{a_1 (\ln n)^2}{n}} \p{1 + \frac{a_2 (\ln n)^2}{n}} \\
    &\le 1 + \frac{(a_1 + a_2)(\ln n)^2}{n} + \frac{a_1 a_2 (\ln  n)^4}{n^2} \\
    &\le 1 + \frac{(a_1 + a_2 + a_1a_2)(\ln n)^2}{n}
  \end{align*}
  for all $n \ge \max(n_1, n_2)$, since $n \ge (\ln n)^2$. A symmetric argument
  holds for the reciprocal. Thus, $f_1(n) f_2(n) \apeqq{a_1+a_2+a_1a_2}{\max(n_1,n_2)}
  g_1(n) g_2(n)$.
\end{proof}
\begin{lemma}
  If $f(n) \apeqq{a}{n_0} g(n)$, then $\frac{1}{f(n)} \apeq
  \frac{1}{g(n)}$.
  \label{lem:reciprocal}
\end{lemma}
\begin{proof}
  \[
    \frac{1/f(n)}{1/g(n)} = \frac{g(n)}{f(n)} \le 1 + \frac{a(\ln n)^2}{n}
  \]
  for $n \ge n_0$. A symmetric argument holds for the reciprocal.
\end{proof}

\begin{lemma}
  If $g_1(n) \le f(n) \le g_2(n)$, $g_1(n) \apeq h(n)$, and $g_2(n) \apeq h(n)$,
  then $f(n) \apeq h(n)$.
  \label{lem:squeeze}
\end{lemma}
\begin{proof}
  \[
    \frac{f(n)}{h(n)} \le \frac{g_2(n)}{h(n)}
  \]
  and
  \[
    \frac{h(n)}{f(n)} \le \frac{h(n)}{g_1(n)},
  \]
  proving the lemma by definition.
\end{proof}

\begin{fact}
  For all $x \ge 1$, $\ln x \le x$ and $(\ln x)^2 \le x$.
\end{fact}

\begin{lemma}
  For $a,b,c,d > 0$, if $\frac{a}{b} \le \frac{c}{d}$, then
  \[
    \frac{a}{b} \le \frac{a+c}{b+d} \le \frac{c}{d}.
  \]
  \label{lem:ratio_bound}
\end{lemma}
\begin{proof}
  Since $\frac{a}{b} \le \frac{c}{d}$, $\frac{d}{b} \le \frac{c}{a}$. Therefore,
  \[
    \frac{a+c}{b+d} = \frac{a}{b} \cdot \frac{1 + c/a}{1+d/b} \ge
    \frac{a}{b} \cdot \frac{1+d/b}{1+d/b} = \frac{a}{b}.
  \]
  Similarly,
  \[
    \frac{a+c}{b+d} = \frac{c}{d} \cdot \frac{1 + a/c}{1+b/d} \le
    \frac{c}{d} \cdot \frac{1+b/d}{1+b/d} = \frac{c}{d}.
  \]
\end{proof}

\begin{lemma}
  \[
    \frac{a-(\ln n)^2/n}{b} \apeq \frac{a}{b}
  \]
  \label{lem:sub}
\end{lemma}
\begin{proof}
  \[
    \frac{\frac{a-(\ln n)^2/n}{b}}{\frac{a}{b}} = 1 - \frac{(\ln n)^2}{n} \le 1
  \]
  \[
    \frac{\frac{a}{b}}{\frac{a-(\ln n)^2/n}{b}} =
    \frac{1}{1-\frac{(\ln n)^2}{an}} = 1 + \frac{\frac{(\ln
    n)^2}{an}}{1-\frac{(\ln n)^2}{an}} \le 1 + \frac{2(\ln n)^2}{an}
  \]
  for $n \ge 16/a^4$.
\end{proof}

\section{Lemmas for Appendix~\ref{app:pl_proofs}}
\begin{lemma}
  For $k \le (1-c) \ln n$,
  \[
    \cdf{n(1+c)-k+1}{n(1+c)}\p{\p{\frac{cn}{\ln n}^{1/(1+\delta)}}} \le
    \frac{\sqrt{k}}{n}.
  \]
  \label{lem:y_cdf_bound}
\end{lemma}
\begin{proof}
  We can write
  \begin{align*}
    \cdf{n(1+c)-k+1}{n(1+c)}\p{\p{\frac{cn}{\ln n}}^{1/(1+\delta)}}
    &= \sum_{j=0}^{k-1} \binom{n(1+c)}{j} \p{1 - \frac{\ln n}{cn}}^{n(1+c)-j}
    \p{\frac{\ln n}{cn}}^j \\
    &\le \sum_{j=0}^{k-1} \frac{(n(1+c))^j}{j!} \exp\p{-\p{\frac{\ln
    n}{cn}}(n(1+c)-j)} \p{\frac{\ln n}{cn}}^j \\
    &= \sum_{j=0}^{k-1} \frac{1}{j!} \p{\frac{(1+c)\ln n}{c}}^j \exp\p{-\ln
    n\p{1 + c^{-1}\p{1-\frac{j}{n}}}} \\
    &= \frac{1}{n} \sum_{j=0}^{k-1} \frac{1}{j!} \p{(1+c^{-1})\ln n}^j
    \p{\frac{1}{n}}^{c^{-1}\p{1-\frac{j}{n}}} \\
    &\le \frac{1}{n^2} + \frac{1}{n} \sum_{j=1}^{k-1}
    \frac{1}{\sqrt{2\pi j}}
    \p{\frac{e(1+c^{-1})\ln n}{j}}^j 
    \p{\frac{1}{n}}^{c^{-1}\p{1-\frac{j}{n}}} \numberthis \label{eq:sum_j_cdf}
  \end{align*}
  by Stirling's approximation. The term
  \[
    \p{\frac{e(1+c^{-1}) \ln n}{j}}^j
  \]
  is increasing whenever it's natural log,
  \[
    j \p{1 + \ln (1+c^{-1}) + \ln \ln n - \ln j},
  \]
  is increasing. This has derivative
  \[
    1 + \ln (1+c^{-1}) + \ln \ln n - \ln j - 1 = \ln(1+c^{-1}) + \ln \ln n - \ln
    j \ge \ln \ln n - \ln j.
  \]
  Thus, it is increasing for $j \le \ln n$. For $j \le (1-c)\ln n$, we have
  \begin{align*}
    \p{\frac{e(1+c^{-1}) \ln n}{j}}^j
    &\le \p{\frac{e(1+c^{-1}) \ln n}{(1-c)\ln n}}^{(1-c)\ln n} \\
    &= \p{\frac{e(1+c^{-1})}{1-c}}^{(1-c)\ln n} \\
    &= \exp\p{1+\ln(1+c^{-1}) - \ln(1-c)}^{(1-c)\ln n} \\
    &= \exp\p{\ln n}^{(1+\ln(1 + c^{-1}) - \ln(1-c))(1-c)} \\
    &= n^{(1+\ln(1 + c^{-1}) - \ln(1-c))(1-c)} \\
    &\le n^{(1 + c^{-1} + c + c^2)(1-c)} \\
    &= n^{1 + c^{-1} + c + c^2 - c - 1 - c^2 - c^3} \\
    &= n^{c^{-1} - c^3} \\
    &\le n^{c^{-1}(1 - j/n)}
  \end{align*}
  for sufficiently large $n$, since $j \le (1-c)\ln n$. Combining this with
  \eqref{eq:sum_j_cdf}, we have
  \begin{align*}
    \cdf{n(1+c)-k+1}{n(1+c)}\p{\p{\frac{cn}{\ln n}}^{1/(1+\delta)}}
    &\le \frac{1}{n^2} + \frac{1}{n\sqrt{2\pi}} \sum_{j=1}^{k-1}
    \frac{1}{\sqrt{j}} \\
    &\le \frac{1}{n^2} + \frac{1}{n\sqrt{2\pi}} \p{1 + \int_1^k
      \frac{1}{\sqrt{j}}\, dj} \\
    &\le \frac{1}{n^2} + \frac{\sqrt{k}}{n\sqrt{2\pi}} \\
    &\le \frac{\sqrt{k}}{n}
  \end{align*}
\end{proof}

\begin{lemma}
  For $y \ge (cn/\ln n)^{1/(1+\delta)}$,
  \[
    \p{1-(\beta y)^{-(1+\delta)}}^{\alpha n} \apeq \p{1-y^{-(1+\delta)}}^{cn},
  \]
  \label{lem:y_approx_c}
\end{lemma}
\begin{proof}
  We know that $1-(\beta y)^{-(1+\delta)} \ge
  (1-y^{-(1+\delta)})^{\beta^{-(1+\delta)}}$ from the Taylor expansion, giving
  us
  \[
    \p{1-(\beta y)^{-(1+\delta)}}^{\alpha n} \ge
    \p{(1-y^{-(1+\delta)})^{\beta^{-(1+\delta)}}}^{\alpha n} =
    (1-y^{-(1+\delta)})^{cn}.
  \]
  On the other hand, for $y \ge (cn/\ln n)^{1/(1+\delta)}$,
  \[
    \p{1-(\beta y)^{-(1+\delta)}}^{\alpha n} \le \exp\p{-c y^{-(1+\delta)} n}
    \le \frac{\p{1-y^{-(1+\delta)}}^{cn}}{1-c n y^{-2(1+\delta)}}
    \le \frac{\p{1-y^{-(1+\delta)}}^{cn}}{1 - \frac{(\ln n)^2}{cn}} \apeq
    \p{1-y^{-(1+\delta)}}^{cn}.
  \]
\end{proof}

\begin{lemma}
  For $k \le ((1-c^2) \ln n)/2$,
  \begin{align*}
    \cdf{\beta^{1+\delta}n(1+c)-j}{\beta^{1+\delta}n(1+c)}\p{\p{\frac{\alpha n}{\ln n}}^{1/(1+\delta)}}
    \le \frac{\sqrt{k}}{n},
  \end{align*}
  \label{lem:x_cdf_bound}
\end{lemma}
\begin{proof}
  We begin with
  \begin{align*}
    \cdf{\beta^{1+\delta}n(1+c)-j}{\beta^{1+\delta}n(1+c)} &\p{\p{\frac{\alpha
    n}{\ln n}}^{1/(1+\delta)}} \\
    &= \sum_{\ell=0}^{j} \binom{\beta^{1+\delta}n(1+c)}{\ell} \p{1 - \frac{\ln
    n}{\alpha n}}^{\beta^{1+\delta}n(1+c)-\ell}
    \p{\frac{\ln n}{\alpha n}}^\ell \\
    &\le \sum_{\ell=0}^{j} \frac{(\beta^{1+\delta}n(1+c))^\ell}{\ell!} \exp\p{-\p{\frac{\ln
    n}{\alpha n}}(\beta^{1+\delta}n(1+c)-\ell)} \p{\frac{\ln n}{\alpha n}}^\ell \\
    &= \sum_{\ell=0}^{j} \frac{1}{\ell!} \p{\frac{\beta^{1+\delta}(1+c)\ln
    n}{c}}^\ell \exp\p{-\ln n\p{1 + c^{-1}\p{1-\frac{\ell}{n}}}}
    \\
    &= \frac{1}{n} \sum_{\ell=0}^{j} \frac{1}{\ell!} \p{\beta^{1+\delta}(1+c^{-1})\ln n}^\ell
    \p{\frac{1}{n}}^{c^{-1}\p{1-\frac{\ell}{n}}} \\
    &\le \frac{1}{n^2} + \frac{1}{n} \sum_{\ell=1}^{j}
    \frac{1}{\sqrt{2\pi \ell}}
    \p{\frac{e\beta^{1+\delta}(1+c^{-1})\ln n}{\ell}}^\ell 
    \p{\frac{1}{n}}^{c^{-1}\p{1-\frac{\ell}{n}}} \numberthis \label{eq:sum_ell_cdf}
  \end{align*}
  We apply a similar argument as in Lemma \ref{lem:y_cdf_bound}, showing that for
  $\ell \le ((1-c^2)/2) \ln n$,
  \begin{align*}
    \p{\frac{e\beta^{1+\delta}(1+c^{-1}) \ln n}{\ell}}^\ell
    &\le \p{\frac{ec^{-1}(1+c^{-1}) \ln n}{(1-c)\ln n}}^{((1-c^2)/2)\ln n} \\
    &= \p{\frac{ec^{-1}(1+c^{-1})}{1-c}}^{(1-c)\ln n} \\
    &= \exp\p{1+ \ln c^{-1} + \ln(1+c^{-1}) - \ln(1-c)}^{((1-c^2)/2)\ln n} \\
    &= \exp\p{\ln n}^{(1+ \ln c^{-1} + \ln(1 + c^{-1}) - \ln(1-c))((1-c^2)/2)} \\
    &= n^{(1+ \ln c^{-1} + \ln(1 + c^{-1}) - \ln(1-c))((1-c^2)/2)} \\
    &\le n^{(1 + (c^{-1}-1) + c^{-1} + c + c^2)((1-c^2)/2)} \\
    &= n^{(2c^{-1} + c + c^2)((1-c^2)/2)} \\
    &\le n^{(2c^{-1} + 2c)((1-c^2)/2)} \tag{$c \le 1$} \\
    &= n^{c^{-1}(1+c^2)(1-c^2)} \\
    &= n^{c^{-1}(1 - c^4)} \\
    &\le n^{c^{-1}(1 - \ell/n)}
  \end{align*}
  for sufficiently large $n$. This gives us
  \begin{align*}
    \cdf{\beta^{1+\delta}n(1+c)-j}{\beta^{1+\delta}n(1+c)}\p{\p{\frac{\alpha n}{\ln n}}^{1/(1+\delta)}}
    &\le \frac{1}{n^2} + \frac{1}{n} \sum_{\ell=1}^j \frac{1}{\sqrt{2\pi \ell}}
    \le \frac{\sqrt{k}}{n},
  \end{align*}
\end{proof}

\begin{lemma}
  For $0 < a < b$ and $c > 0$,
  \[
    \frac{a+c}{b+c} > \frac{a}{b}
  \]
  \label{lem:frac_bigger}
\end{lemma}
\begin{proof}
  \[
    \frac{a+c}{b+c}
    = \frac{a(1+c/a)}{b(1+c/b)}
    > \frac{a(1+c/b)}{b(1+c/b)}
    = \frac{a}{b}
  \]
\end{proof}

\begin{lemma}
  For $0 \le y \le a_1 \cdot \frac{\ln n}{n}$ and $|z| \le a_2 \ln n$,
  \[
    |(1-y)^z - (1-yz)| = O\p{\frac{1}{n}}
  \]
  \label{lem:taylor_1_yz}
\end{lemma}
\begin{proof}
  By Taylor's theorem,
  \[
    f(y) = (1-y)^z = 1 - yz \pm \frac{f''(\varepsilon)}{2} y^2
  \]
  for some $0 \le \varepsilon \le y$. Note that
  \[
    f''(\varepsilon) = z(z-1) (1-\varepsilon)^{z-2} \le |z(z-1)|
    \exp(-\varepsilon(z-2)) \le |z(z-1)| \exp(\varepsilon|z-2|).
  \]
  Since $\varepsilon \le y \le a_1 \cdot \frac{\ln n}{n}$ and $|z| \le a_2 \ln
  n$,
  \[
    |z(z-1)| \exp(\varepsilon|z-2|) \le a_2^2 (\ln n)^2 n^{-2} n^{a_1|z-2|/n}.
  \]
  This gives us
  \[
    \frac{f''(\varepsilon)}{2} y^2 \le \frac{a_2^2 (\ln n)^2 n^{-2}
    n^{a_1|z-2|/n}}{2} a_1^2 (\ln n)^2 n^{-2}
    \le a_1^2 a_2^2 (\ln n)^4 n^{-(2 - a_1(a_2\ln n+2)/n)}.
  \]
  Using $\ln n = n^{\ln \ln n/\ln n}$, this is
  \[
    \frac{a_1^2 a_2^2}{n} n^{-(1-a_1(a_2 \ln n + 2)/n - 4\ln \ln n/\ln n)}.
  \]
  For sufficiently large $n$, $a_1(a_2 \ln n + 2)/n + 4\ln \ln n/\ln n \le 1$, so
  \[
    \frac{a_1^2 a_2^2}{n} n^{-(1-a_1(\ln n + 2)/n - 4\ln \ln n/\ln n)} \le
    \frac{a_1^2 a_2^2}{n} = O(1/n),
  \]
  which proves the lemma.
\end{proof}
\begin{lemma}
  \[
    \cdf{a n}{a n}\p{b \p{\frac{c n}{\ln n}}^{1/(1+\delta)}}
    \le n^{-ab^{-(1+\delta)}/c}
  \]
  \label{lem:cdf_bound}
\end{lemma}
\begin{proof}
  \begin{align*}
    \cdf{a n}{a n}\p{b \p{\frac{c n}{\ln n}}^{1/(1+\delta)}}
    &= \p{1 - b^{-(1+\delta)} \frac{\ln n}{c n}}^{a n} \\
    &\le \exp\p{-\frac{ab^{-(1+\delta)}}{c}\ln n} \\
    &= n^{-ab^{-(1+\delta)}/c}
  \end{align*}
\end{proof}

\begin{lemma}
  \[
    \E{\zlnn} \ge \p{\frac{Cn}{\ln n}}^{1/(1+\delta)}
  \]
  for $C \ge 1$ and sufficiently large $n$.
  \label{lem:zlnn}
\end{lemma}
\begin{proof}
  \begin{align*}
    \E{\zlnn} &= \E{\znC} \prod_{j=1}^{\ln n - 1} \p{1 - \frac{1}{(1+\delta)j}} \\
    &= \E{\znC} \prod_{j=1}^{\ln n - 1} \p{\frac{(1+\delta)j - 1}{(1+\delta)j}} \\
    &= \E{\znC} \prod_{j=1}^{\ln n - 1} \p{\frac{j - 1/(1+\delta)}{j}} \\
    &= \E{\znC} \frac{\Gamma(\ln n - 1/(1+\delta))}{\Gamma(\delta/(1+\delta))
    \Gamma(\ln n)} \\
    &\ge \Gamma\p{\frac{\delta}{1+\delta}} (Cn)^{1/(1+\delta)}
    \frac{\Gamma(\ln n - 1/(1+\delta))}{\Gamma(\delta/(1+\delta))
    \Gamma(\ln n)} \tag{by Lemma \ref{lem:os_approx}} \\
    &= (Cn)^{1/(1+\delta)} \frac{\Gamma(\ln n - 1/(1+\delta))}{\Gamma(\ln
    n)} \\
    &\ge \p{\frac{Cn}{\ln n}}^{1/(1+\delta)} \p{1 + \frac{\frac{1}{1+\delta}
    \cdot \frac{1+2\delta}{1+\delta}}{\ln n} - O\p{\frac{1}{(\ln n)^2}}} \tag{by
    \cite{gamma-ratio}} \\
    &\ge \p{\frac{Cn}{\ln n}}^{1/(1+\delta)} \tag{for sufficiently large
    $n$}
  \end{align*}
\end{proof}

\begin{lemma}
  For $k = O(\ln n)$,
  \[
    \frac{\binom{n}{k}}{\binom{n(1+c)}{k}} \apeq \frac{1}{(1+c)^k}.
  \]
  \label{lem:binom_ratio}
\end{lemma}
\begin{proof}
  \[
    \frac{\binom{n}{k}}{\binom{n(1+c)}{k}} = \frac{n(n-1) \cdots
    (n-k+1)}{n(1+c) (n(1+c)-1) \cdots (n(1+c) - k+1)}.
  \]
  Each term $(n-j)/(n(1+c)-j)$ is between $1/(1+c)$ and $(1-j/n)/(1+c)$.
  Therefore, the entire product is at least
  \[
    \prod_{j=0}^{k-1} \frac{1}{1+c} \p{1-\frac{j}{n}} = \frac{1}{(1+c)^k}
    \prod_{j=0}^{k-1} \p{1-\frac{j}{n}} \ge \frac{1}{(1+c)^k}
    \p{1-\frac{k^2}{n}}
  \]
  and at most $1/(1+c)^k$.
  This means that
  \[
    \frac{1}{(1+c)^k} \ge \frac{\binom{n}{k}}{\binom{n(1+c)}{k}} \ge
    \frac{1}{(1+c)^k} \p{1 - \frac{(\ln n)^2}{n}} \apeq \frac{1}{(1+c)^k}
  \]
\end{proof}

\begin{lemma}
  For $0 < z < 1$, and $y \ge 0$,
  \[
    (1+y)^z < 1 + yz.
  \]
  \label{lem:first_expansion}
\end{lemma}
\begin{proof}
  Let $w = z^{-1}$. Then, the lemma is true if and only if for $w > 1$,
  \[
    1 + y < \p{1 + \frac{y}{w}}^w.
  \]
  Note that for $w = 1$, we have equality. We will show that the function
  \[
    f(w) = \p{1 + \frac{y}{w}}^w
  \]
  has nonnegative derivative for $w \ge 1$. This is equivalent to showing the
  same for its log, which is
  \begin{align*}
    \frac{d}{dw} \log f(w) &= \frac{d}{dw} w \log\p{1 + \frac{y}{w}} \\
    &= \log \p{1 + \frac{y}{w}} + \frac{w}{1 + \frac{y}{w}} \cdot
    \p{-\frac{y}{w^2}} \\
    &= \log \p{1 + \frac{y}{w}} - \frac{\frac{y}{w}}{1 + \frac{y}{w}}
  \end{align*}
  Let $x = 1 + \frac{y}{w}$. Then, the lemma is true if for $x > 1$,
  \begin{align*}
    \log(x) - \frac{x-1}{x} &>0 \\
    x\log(x) &> x-1
  \end{align*}
  Both are $0$ at $x=1$, but the left hand side has derivative $1+\log(x)$ while
  the right hand side has derivative $1$, so left hand side will be strictly
  larger than the right hand side for $x > 1$.
\end{proof}

\begin{lemma}
  \[
    \E{\zm} \apeq \Gamma\p{\frac{\delta}{1+\delta}} m^{1/(1+\delta)}.
  \]
  Also,
  \[
    \E{\zm} \ge \Gamma\p{\frac{\delta}{1+\delta}} m^{1/(1+\delta)}.
  \]
  \label{lem:os_approx}
\end{lemma}
\begin{proof}
  From \cite{expected-ratio}, we have
  \[
    \E{\zm} = \frac{\Gamma(m+1) \Gamma\p{1-\frac{1}{1+\delta}}}{\Gamma\p{m +
    \frac{\delta}{1+\delta}}}.
  \]
  By \cite{gamma-ratio},
  \[
    \frac{\Gamma(m+1)}{\Gamma\p{m + \frac{\delta}{1+\delta}}} = m^{1/(1+\delta)}
    \p{1 + \frac{\p{\frac{1}{1+\delta}} \p{\frac{\delta}{1+\delta}}}{2m} +
    O\p{\frac{1}{m^2}}} \ge m^{1/(1+\delta)}
  \]
  This means
  \[
    \Gamma\p{\frac{\delta}{1+\delta}} m^{1/(1+\delta)} \le \E{\zm} \le
    \Gamma\p{\frac{\delta}{1+\delta}} m^{1/(1+\delta)} \p{1 +
    O\p{\frac{1}{n}}},
  \]
  so
  \[
    \Gamma\p{\frac{\delta}{1+\delta}} m^{1/(1+\delta)} \apeq \E{\zm}.
  \]
\end{proof}

\begin{lemma}[\cite{expected-ratio}, Formula 1]
  \[
    \E{\zmk} = \p{1 - \frac{1}{k(1+\delta)}}\E{\zmko}
  \]
  \label{lem:recurrence}
\end{lemma}

\section{Lemmas and Proofs for Section~\ref{sec:bounded}} \label{app:missing_bd}
\label{app:bounded_proofs}
\begin{proof}[Proof of Theorem~\ref{thm:bounded}]
  To proceed, we need some notation. Let $L$ be the event that $\xnn \ge
  T \cap \ynn \ge T$ (the samples are ``large''). Let $G$ be the event that $b(\xn) <
  \ynn$, meaning $G$ is the event that the policy has an effect. Let $D$ be the
  random variable $\xn - \ynn$. We want to show that $\E{D|G} > 0$. To do so, we
  observe that by Lemma \ref{lem:da}, is sufficient to
  show that $\E{D|L} > \frac{\Prb{\overline{L}}}{\Prb{L}}$.
  By Lemma \ref{lem:lsmall}, we know that $\Prb{\overline{L}} \le 2n F(T)^{n-1}$.
  To complete the proof, we need to show that $\E{D|L}$ is large, which we do via
  Lemma~\ref{lem:edl}.

  Since $\Prb{L} \ge 1 - 2nF(T)^{n-1}$, there exists $N_1$ such that
  for all $n \ge N_1$, $\Prb{L} \ge 1/2$. Using Lemma \ref{lem:edl}, if $n
  \ge N_1$, it is sufficient to have
  \begin{align*}
    \E{D|L} &> \frac{\Prb{\overline{L}}}{\frac{1}{2}} \\
    K(F(T) + \eta)^{n-1} &> 4nF(T)^{n-1} \\
    \p{1 + \frac{\eta}{F(T)}}^n &> \frac{4n}{K} \\
    n\log\p{1 + \frac{\eta}{F(T)}} &> \log n + \log\p{\frac{4}{K}} \\
    \sqrt{n}\log\p{1 + \frac{\eta}{F(T)}} &> 2
    \tag{$n \ge 4/K$, using $\sqrt{n} > \log n$} \\
    n &> 4\p{\log\p{1 + \frac{\eta}{F(T)}}}^{-2} = N_2
  \end{align*}

  Thus, for $n > \max\{N_1, N_2, 4/K\}$, $\E{D|L} >
  \frac{\Prb{\overline{L}}}{\Prb{L}}$, which by Lemma \ref{lem:edl} implies that
  $\E{D|G} > 0$. This completes the proof of Theorem \ref{thm:bounded}.
\end{proof}
\begin{lemma} \label{lem:da}
  If $L \Rightarrow G$ and $D \ge -1$, then
  $\E{D|L} > \frac{\Prb{\overline{L}}}{\Prb{L}}$ implies $\E{D|G} > 0$.
\end{lemma}
\begin{proof}
  \begin{align*}
    \E{D|G} &= \E{D \cdot \ind{L} | G} + \E{D \cdot \ind{\overline{L}} | G} \\
    &= \frac{\E{D \cdot \ind{L} \cdot \ind{G}} + \E{D \cdot \ind{\overline{L}} \cdot
    \ind{G}}}{\Prb{G}} \\
    &= \frac{\E{D \cdot \ind{L}} + \E{D \cdot \ind{\overline{L}} \cdot
    \ind{G}}}{\Prb{G}} \tag{$L \Rightarrow G$} \\
    &\ge \frac{\E{D \cdot \ind{L}} - \E{\ind{\overline{L}} \cdot
    \ind{G}}}{\Prb{G}} \tag{$D \ge -1$} \\
    &\ge \frac{\E{D \cdot \ind{L}} - \E{\ind{\overline{L}}}}{\Prb{G}}
    \tag{$\ind{G} \le 1$} \\
    &= \frac{\E{D|L} \Prb{L} - \Prb{\overline{L}}}{\Prb{G}}
  \end{align*}
  \begin{align*}
    \frac{\E{D|L} \Prb{L} - \Prb{\overline{L}}}{\Prb{G}} &> 0 \\
    \iff \E{D|L} \Prb{L} - \Prb{\overline{L}} &> 0 \\
    \iff \E{D|L} &> \frac{\Prb{\overline{L}}}{\Prb{L}}
  \end{align*}
\end{proof}
\begin{lemma} \label{lem:lsmall}
  For $\xnn$, $\ynn$ order statistics from a distribution with support on
  $[0,1]$,
  \[
    \Prb{\xnn \ge T \cap \ynn \ge T} \le 2n F(T)^{n-1}.
  \]
\end{lemma}
\begin{proof}
  \begin{align*}
    \Prb{\xnn \ge T \cap \ynn \ge T}
    &= \Prb{\xnn \ge T} \Prb{\ynn \ge T} \\
    &= (1-F_{(n-1)}(T))^2 \\
    &= (1 - n F(T)^{n-1} (1-F(T)) -
    F(T)^n)^2 \\
    &= (1 - n F(T)^{n-1} + (n-1)F(T)^n)^2 \\
    &\ge (1 - n F(T)^{n-1})^2 \\
    &\ge 1 - 2nF(T)^{n-1}
  \end{align*}
  \[
    \Prb{\overline{L}} = 1 - \Prb{L} \le 2n F(T)^{n-1}
  \]
\end{proof}

\begin{lemma} \label{lem:edl}
  There exist constants $\eta > 0$ and $K > 0$ such that $\E{D|L} \ge
  K(F(T) + \eta)^{n-1}$
\end{lemma}
\begin{proof}
  First, let $f_Z$ and $F_Z$ be the pdf and cdf respectively of $Y | Y \ge
  T$, i.e. $F_Z(x) = \frac{F(x) -
    F(T)}{1-F(T)}$ and $f_Z = F_Z'$. Note that
  \begin{align*}
    \E{D|L} &= \E{\xn - \ynn | L} \\
    &= \E{\xn | \xnn \ge T} - \E{\ynn | \ynn \ge
    T} \\
    &= \E{\yn | \ynn \ge T} - \E{\ynn | \ynn \ge T} \\
    &= \E{\yn - \ynn | \ynn \ge T} \\
  \end{align*}
  Let $M$ be a random variable corresponding to the number of samples from $Y_1,
  \dots, Y_n$ that are larger than $T$. We can rewrite this as
  \begin{align*}
    \E{D|L} &= \sum_{m=2}^M \E{\yn - \ynn | \ynn \ge T, M=m}
    \Prb{M=m|\ynn \ge T} \\
    &= \sum_{m=2}^M \E{\yn - \ynn | M=m} \Prb{M=m|\ynn \ge T}
    \tag{$M \ge 2 \Longrightarrow \ynn \ge T$}
  \end{align*}
  Conditioning on $M = m$, $\yn$ and $\ynn$ have the same distributions as $\zm$
  and $\zmm$ respectively, where $\zk$ is the $k$th order statistic of
  random variables $Z_1, Z_2, \dots, Z_m$ drawn from the distribution with cdf
  $F_Z$. We will use $\cdfz{k}{m}$ to denote the cdf of $\zk$. Thus, $\E{\yn -
  \ynn | M = m} = \E{\zm - \zmm}$. Using an analysis similar to that of
  \cite{expected-diff},
  \begin{align*}
    \E{\zm - \zmm} &= \int_{T}^1 (1-\cdfz{m}{m}(x)) -
    (1-\cdfz{m-1}{m}(x)) \, dx \\
    &= \int_{T}^1 \cdfz{m-1}{m}-\cdfz{m}{m}(x) \, dx \\
    &= \int_{T}^1 \binom{m}{m-1} F_Z(x)^{m-1} (1-F_Z(x)) \, dx \\
    &\ge \int_{T}^1 F_Z(x)^{m-1} (1-F_Z(x)) \, dx
  \end{align*}
  Choose $\eta \in (0, 1-F(T))$ and $\eta' \in (\eta,
  1-F(T))$. Let $r = F_Z^{-1}(F(T) + \eta)$ and $r' =
  F_Z^{-1}(F(T) + \eta')$. Note that $T < r < r' < 1$ because
  otherwise $F_Z$ would have infinite slope at $r$ or $r'$, which is impossible
  because $f_Z$ is continuous over a compact set and therefore has a finite
  maxmium. Moreover, it must be the case that $F(T) < 1$
  because by assumption, $\sup_{x : f(x) > 0} = 1$. If $F(T)$ were $1$, this
  would imply that $\sup_{x : f(x) > 0} = T < 1$, which is a contradiction.
  \begin{align*}
    \int_{T}^1 F_Z(x)^{m-1} (1-F_Z(x)) \, dx &\ge
    \int_{r}^1 F_Z(x)^{m-1} (1-F_Z(x)) \\
    &\ge \int_{r}^1 F_Z(r)^{m-1} (1-F_Z(x)) \\
    &= (F(T) + \eta)^{m-1} \int_r^1 1-F_Z(x) \, dx \\
    &\ge (F(T) + \eta)^{n-1} \int_r^1 1-F_Z(x) \, dx \\
    &\ge (F(T) + \eta)^{n-1} \int_r^{r'} 1-F_Z(x) \, dx \\
    &\ge (F(T) + \eta)^{n-1} \int_r^{r'} 1-F_Z(r') \, dx
    \tag{$F_Z(x) \le F_Z(r')$ for $x \le r'$} \\
    &= (F(T) + \eta)^{n-1} (r'-r) (1-(F(T) + \eta')) \\
    &= (F(T) + \eta)^{n-1} [F_Z^{-1}(F(T)+\eta') -
    F_Z^{-1}(F(T) + \eta)] (1-F(T) - \eta') \\
    &= K(F(T) + \eta)^{n-1}
  \end{align*}
  where $K = [F_Z^{-1}(F(T)+\eta') - F_Z^{-1}(F(T) + \eta)] (1-F(T) -
  \eta')$.
  Since this is independent of $m$, we have
  \begin{align*}
    \E{D|L} &= \sum_{m=2}^n \E{\yn - \ynn | M=m} \Prb{M=m|\ynn \ge
      T} \\
      &\ge \sum_{m=2}^n K(F(T) + \eta)^{n-1} \Prb{M=m|\ynn \ge
        T} \\
        &= K(F(T) + \eta)^{n-1}
  \end{align*}
\end{proof}

\end{document}